\documentclass[journal]{IEEEtran}
\usepackage[bf,font=normalsize]{caption}

\usepackage{float}
\usepackage{graphicx}
\usepackage{subfig}
\usepackage{amsmath}
\usepackage{amsfonts}
\usepackage{amsthm}
\usepackage{caption}
\usepackage{multicol}
\usepackage{lipsum}
\usepackage{pgfplots}
\usepackage{verbatim}
\usepackage{cleveref}
\usepackage{booktabs}
\usepackage{listings}
\usepackage{array}
\usepackage{soul}

\usepackage[inline]{enumitem}

\newsavebox{\measurebox}
\theoremstyle{definition}
\newtheorem{definition}{Definition}[section]
\newtheorem{theorem}{Theorem}
\newtheorem{remark}{Remark}
\newtheorem{problem}{Problem}

\lstdefinelanguage{JavaScript}{
  keywords={typeof, new, true, false, catch, function, return, null, catch, switch, var, if, in, while, do, else, case, break},
  keywordstyle=\color{blue}\bfseries,
  ndkeywords={class, export, boolean, throw, implements, import, this},
  ndkeywordstyle=\color{darkgray}\bfseries,
  identifierstyle=\color{black},
  sensitive=false,  
  comment=[l]{//},
  morecomment=[s]{/*}{*/},
  commentstyle=\color{purple}\ttfamily,
  stringstyle=\color{black}\ttfamily,
  morestring=[b]',
  morestring=[b]",
  showstringspaces=false
}

\ifCLASSINFOpdf
\else
\fi
\newcommand{\LMC}[1]{\textcolor{black}{#1}}

\begin{document}
%
% paper title
% Titles are generally capitalized except for words such as a, an, and, as,
% at, but, by, for, in, nor, of, on, or, the, to and up, which are usually
% not capitalized unless they are the first or last word of the title.
% Linebreaks \\ can be used within to get better formatting as desired.
% Do not put math or special symbols in the title.
\title{Cooperative Intersection Crossing over 5G}
\author{Luca~Maria~Castiglione,
        Paolo~Falcone,
        Alberto~Petrillo,
        Simon~Pietro~Romano,
        and~Stefania~Santini
        % <-this % stops a space
\thanks{L.M.~Castiglione is with Department of Computing, Imperial College London. Work done whilst at DIETI.}%
\thanks{A.~Petrillo, S.P.~Romano and S.~ Santini are with the Department of Electrical Engineering and Information Technology (DIETI), University of Napoli Federico II, Napoli}% <-this % stops a space
\thanks{P.Falcone is with Department of Electrical Engineering, Chalmers University of Technology, Gothenburg}}% <-this % stops a space

%\markboth{IEEE/ACM Transactions on Networking}%
%{Castiglione \MakeLowercase{\textit{et al.}}: Improving autonomous vehicles coordination}
% The only time the second header will appear is for the odd numbered pages
% after the title page when using the twoside option.
% 
% *** Note that you probably will NOT want to include the author's ***
% *** name in the headers of peer review papers.                   ***
% You can use \ifCLASSOPTIONpeerreview for conditional compilation here if
% you desire.
% If you want to put a publisher's ID mark on the page you can do it like
% this:
%\IEEEpubid{0000--0000/00\$00.00~\copyright~2015 IEEE}
% Remember, if you use this you must call \IEEEpubidadjcol in the second
% column for its text to clear the IEEEpubid mark.
% use for special paper notices
%\IEEEspecialpapernotice{(Invited Paper)}
% make the title area
\maketitle

\begin{abstract}
Autonomous driving is a safety critical application of sensing and decision-making technologies. Communication technologies extend the awareness capabilities of vehicles, beyond what is achievable with the on-board systems only. Nonetheless, issues typically related to wireless networking must be taken into account when designing safe and reliable autonomous systems. The aim of this work is to present a control algorithm and a communication paradigm over $5$G networks for negotiating traffic junctions in urban areas. \LMC{The proposed control framework has been shown to converge in a finite time and the supporting communication software has been designed with the objective of minimizing communication delays. At the same time, the underlying network guarantees reliability of the communication.} The proposed framework has been successfully deployed and tested, in partnership with Ericsson AB, at the AstaZero proving ground in Goteborg, Sweden. In our experiments, three autonomous vehicles successfully drove through an intersection of $235$ square meters in a urban scenario. 
%Obtained results are disclosed and analyzed in this paper.
%  The objective is to minimize  
%\PF{The proposed communication and control framework is designed with the objective of minimizing} both communication delay and information loss, while at the same time, maximizing the reliability of the communication. The proposed framework %we present 

\end{abstract}

% Note that keywords are not normally used for peerreview papers.

% For peer review papers, you can put extra information on the cover
% page as needed:
% \ifCLASSOPTIONpeerreview
% \begin{center} \bfseries EDICS Category: 3-BBND \end{center}
% \fi
%
% For peerreview papers, this IEEEtran command inserts a page break and
% creates the second title. It will be ignored for other modes.
\IEEEpeerreviewmaketitle

\section{Introduction}
\label{sec:introduction}
% The very first letter is a 2 line initial drop letter followed
% by the rest of the first word in caps.
% 
% form to use if the first word consists of a single letter:
% \IEEEPARstart{A}{demo} file is ....
% 
% form to use if you need the single drop letter followed by
% normal text (unknown if ever used by the IEEE):
% \IEEEPARstart{A}{}demo file is ....
% 
% Some journals put the first two words in caps:
% \IEEEPARstart{T}{his demo} file is ....
% 
% Here we have the typical use of a "T" for an initial drop letter
% and "HIS" in caps to complete the first word.
Autonomous Driving (AD) is definitely one of the most challenging safety critical applications as it involves, among others, advanced sensing and control technologies. %A crucial part of any autonomous driving system is inter-vehicle communication, since
Furthermore, communication with other vehicles and/or the traffic infrastructure is expected to influence the development of AD technologies, as it allows to potentially %dramatically 
improve the environment awareness beyond the range of the current sensing systems such as cameras, lidars and radars. When relying upon network-based coordination, issues typically related to wireless communication must be taken into account in order to design control algorithms and driving software that are guaranteed to be both safe and reliable. Structures such as buildings and walls that are commonly part of urban scenarios act as obstacles against the direct communication between two or more mobile nodes and may significantly degrade the communication performance, due to the so called \emph{shadowing effect}~\cite{ShadowingEffect}. Therefore, it is beneficial to use a communications technology that overcomes the shadowing and implements inter-node communication through an upstream centralized dispatching layer. In this paper we propose to realize such a higher-level dispatching framework by leveraging $5$G-enabled cloud-based inter-vehicle communication. 
With the proposed approach, every vehicle receives real time traffic updates from the cloud and is made aware of the presence of other nodes. %If one is able to impose real time constraints on the communication with the cloud, packet losses and delays %are extremely can be minimized. 
A $5$G-based communications solution overcomes the problems due to local obstacles since the cellular network, with the aid of the base station, can establish reliable communications among vehicles. Such a solution would be cumbersome to deploy and would have high costs if implemented by leveraging classic Wi-Fi 802.11P. This is due to the need of installing numerous access points to overcome shadowing and radio coverage issues. In addition, the cloud based technology powered by Ericsson enables communication to meet stringent time constraints requested by real-time distributed communication algorithms. %Such configuration cannot be deployed on a legacy $4$G network due to the stringent time constraints requested by the real-time algorithms for the autonomous driving of the connected vehicles, which are involved via the distributed control logic. 
The prerequisite for the applicability of fully distributed control architectures for cooperative driving in urban scenarios is the availability of a reliable network that supports $4$G+ cloud based communications. The pre-5G proof-of-concept (PoC) at Astazero uses LTE radio with 5G EPC (\emph{Evolved Packet Core}) and is designed to support low latency ultra-reliable communications. It is an ideal candidate for safety critical autonomous driving applications. \\
In this paper, we address the challenging problem of coordinating connected self-driving cars at urban traffic junctions, where traffic efficiency has to be achieved while guaranteeing safety. In the control paradigm, the cyber-physical system is intuitively represented as a multi-agent system (MAS) composed of different dynamical agents, i.e., the vehicles, that automatically control their dynamical behavior by leveraging both local information and information shared with their neighbors via the communication network. The fully autonomous coordination of the self-driving cars at road intersection is solved by  proposing  a  distributed  nonlinear cooperative protocol based on the MAS abstraction. Note that MAS tools appear to be an alternative viable framework for controlling vehicles in a completely distributed fashion with a computational load compatible with real-life automotive applications.
More notably, the effectiveness of the theoretical framework is experimentally tested by enabling communication through the pre-$5G$ PoC network \LMC{deployed} %installed 
by Ericsson at the \emph{AstaZero} proving ground for autonomous vehicles. Experiments were carried out on two  cars, namely Volvo Car $XC90$ and Volvo Car $S90$, and one truck, Volvo $FH16$.
The main outcome of this research work shows that 4G+/5G networks will definitely play an important role in automotive applications, by allowing safe, real-time and reliable autonomous driving maneuvers.\\ %with this communication on top.\\ 
%Enhancements with respect to the state of the art can be summarized as: {\bf TO DO}.\\
\LMC{The paper is organized as follows. Section \ref{sec:background} presents state of the art strategies for the coordination of autonomous vehicles over street junctions. Section \ref{preliminaries} and \ref{sec:problemstatmenet} introduce, respectively, mathematical preliminaries and the mathematical formulation leveraged to properly describe the tackled problem. Then, in section \ref{sec:ControlProtocol} the adopted control strategy is outlined. The application module developed to enable the communication between vehicles over cellular networks is introduced in section \ref{sec:communicationsoftware}, while section \ref{sec:experimentalsetup} presents a detailed description of the hardware instrumentation deployed on the autonomous vehicles and involved in field trials. Experimental results are  disclosed and validated in Section \ref{sec:experiments}, while network performance measured during the experiments is discussed in section \ref{sec:netperformances}. Finally, conclusions are summarized in section \ref{sec:conclusions}}. %Finally, the paper can be outlined as follows:\\ {\bf TO DO}

\begin{table}[!t]\label{5gCaratteristiche}
\centering
\caption{LTE+ 5G EPC communication against Wi-Fi in urban areas}
\label{tab:wifivs5g}
	\begin{tabular}{|l|p{18mm}|p{18mm}|p{18mm}|}
		\hline
		  & \textbf{Wi-Fi 802.11a} & \textbf{Wi-Fi 802.11p} 	& \textbf{Cellular LTE}\\ \hline
		  \textbf{Shadowing} & Suffers & Suffers & Does not suffer\\ \hline
		  \textbf{Available bands} & $2.4$ GHz and $5$ GHz & $5.85-5.925$ GHz & $400$MHz to $60$ GHz\\ \hline
		  \textbf{Authentication} & Association per station & Originally absent at MAC Level & Association per handover\\ \hline
		  \textbf{Reliability} & Absent & Guaranteed by halving the bandwidth of 802.11a & Ultra Reliable Low Latency Communication (URLLC) supported in 5G NR\\ \hline
		  \textbf{Sustainability} & New installations required & New installations required & Already deployed equipment usable\\ \hline
	\end{tabular}   
\end{table}

\section{Background}
\label{sec:background}
In the rich technical literature about connected autonomous vehicles, different techniques for safe intersection crossing have been mainly categorized as either \textit{centralized} or \textit{decentralized} (see \cite{rios2017survey} and references therein). In centralized approaches, an Intersection Coordination Unit (ICU) acts as a supervisor that %globally 
coordinates vehicles' tasks in order to %minimize the risk of collisions and/or the travel time 
optimize some performance index while avoiding collisions \cite{zhu2015linear,kamal2015vehicle}.
However, when considering an intersection involving a large number of autonomous vehicles, such centralized architectures may result unsuitable because of both their limited capability to gather and process a large data set, and the difficulty arising from solving in real-time the consequent large-scale optimization problem \cite{zheng2017distribute}.
%Differently, decentralized approaches, where each vehicle determines its dynamic behavior on the basis of only information received from its neighbors, implement at single vehicle level decision making algorithms allowing to negotiate the access to the traffic intersection. 
On the other hand, in decentralized approaches each vehicle determines its dynamic behavior on the basis of only the information received by its neighbors. In particular, once the crossing time or order is scheduled, a control strategy locally provides the required acceleration/deceleration profile for each vehicle, based on the information received from its neighboring vehicles.
Optimal control approaches are common to enforce the hard safety constraints necessary to avoid collisions, as for example in \cite{liu2017distributed,chalmersexp1,chalmersexp2}. More notably, \cite{chalmersexp1,chalmersexp2} also 
carried out an experimental campaign by leveraging Vehicle-to-Vehicle (V2V) over Wi-Fi (based on the IEEE 802.11p protocol) and provided the experimental validation of the proposed optimal control approach. In this case, experiments are performed in an extra-urban area where no structure, such as buildings and walls, are present. Therefore, some of the issues related to wireless communications in urban scenarios (see Table~\ref{tab:wifivs5g}) have not been considered. Main limitations in the use of the Wi-Fi Point to Point communication in urban scenarios come from the shape of the latter. To this extent, elements such as buildings, trees and walls constitute an obstacle to high frequency (Wi-Fi) communications by shadowing the signal. To this extent, 802.11p happens to be more suitable in modern cities than 802.11a. In this specification the cyclic prefix length is doubled by halving the bandwidth, which in turn gives more resilience to large delay spreads. However, 4G+/5G is preferable over Wi-FI, in this context, as it has been proved to be more reliable, as well as more sustainable in terms of deployment. Cellular communications are preferred since 5G NR enables reliability through its support to Ultra Reliable Low Latency Communication (URLLC). Also, by leveraging the cellular network for inter-vehicle communication, there is no need to deploy and install further technical equipment on the ground. 4G+/5G is intrinsically less disturbed as cellular networks operate within a controlled and licensed spectrum. 

In view of the above considerations, this work explores the possible use of a pre-$5$G PoC, using LTE radio with 5G EPC technology, for the Cooperative Intersection Crossing (CIC). This entails that each vehicle autonomously makes decisions based on the information it receives from the pre-$5$G network, minimizing the computational delays at the road infrastructure side, that might significantly increase when several vehicles/nodes are approaching the intersection area. Connections between autonomous vehicles and infrastructure have been organized as shown Figure \ref{fig:pp_block_diagram}.

\begin{figure}[t]
\centering
\includegraphics[width=0.99\columnwidth]{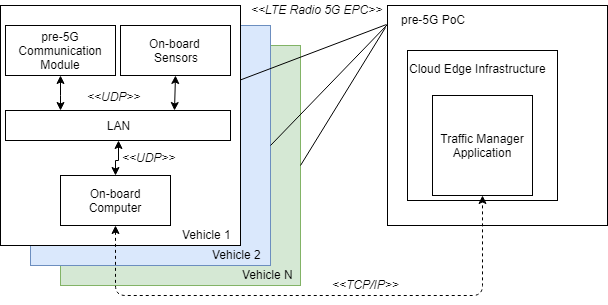}
\caption{Connection among road users and Ericsson pre-5G PoC infrastructure}
\vspace{-10pt}
\label{fig:pp_block_diagram}
\end{figure}

In the following, we theoretically and experimentally prove that the proposed approach is capable to meet hard-enough real-time constraints.
The completely distributed nonlinear finite-time control strategy allows the cooperative negotiation of an intersection, while collisions are prevented by the achievement of the desired virtual formation in a finite time $T$ before the first vehicle accesses the core of the intersection. It is worth noting that, while a cross intersection is considered throughout this paper, the proposed framework can be applied to any type of traffic junction.

\section{Nomenclature and Mathematical Preliminaries}
\label{preliminaries}
%{\bf STEFANIA}
The communication network established among vehicles can be modeled as a graph, where each vehicle is represented by a \textit{node}, while the existence of a communication link between a pair of vehicles by an \textit{edge}.
Specifically, the communication topology of a group of $N$ vehicles can be described by an undirected graph $\mathcal{G}_N = (\mathcal{V}_N,\mathcal{E}_N)$ of order $N$, with vertex set $\mathcal{V}_N = \{ 1,\dots,N \}$ and edge set $\mathcal{E}_N \subset \mathcal{V}_N \times \mathcal{V}_N $, where the presence of the edge $(i,j) \in \mathcal{E}_N$ indicates that the vehicle $i$ receives information from vehicle $j$, and viceversa.
The topology of the graph is associated to the \textit{binary adjacency matrix} $\mathcal{A}_N = [a_{ij}]_{N\times N}$ encoding vehicle communication relationship, where $a_{ij}=1$ if $(i,j) \in \mathcal{E}_N$, and $a_{ij}=0$ otherwise. Note that, $a_{ii}=0$ since self-edges $(i,i)$ are not considered. Therefore, each vehicle $i$ receives the status of all vehicles that are  members of its neighboring set  $\mathcal{N}_i = \{ j\in \mathcal{V}_N: (i,j)\in \mathcal{E}_N, j \neq i \}$.
Moreover, a \textit{path} in a graph is an ordered sequence of vertices such that any pair of consecutive vertices in the sequence is an edge of the graph.
Here, according to the above definitions, the %undirected
graph $\mathcal{G}_N$ that describes the communication topology of the cooperative vehicles is assumed to be connected, although not completely.\\
Next, we introduce definitions and recall results from literature that will be exploited in the manuscript to establish our main results.
%For any vector $a=[a_{1}, a_{2}, \cdots, a_{n}]^{\top}$ and $b=[b_{1, b_{2}, \cdots, b_{n}}]^{\top}$, denote
%\begin{equation}\label{element-wise_product}
%a  \odot b=  [a_{1}b_{1}, a_{2}b_{2}, \cdots, a_{n}b_{n}]^{\top}.
%\end{equation}
\begin{definition}
\label{conn}
(Graph Connectivity)~\cite{chartrand2006introduction}.
An undirected graph $\mathcal{G}_N$ is said to be \textit{connected} if there exists a path between any two vertices.
In addition, if there exists a path from any vertex to any other vertex, the $\mathcal{G}_N$ is said to be \textit{completely connected}.
\end{definition}

\begin{definition}
\label{sig}(Sig Function)~\cite{lu2013finite,bhat2000finite}. Let
\begin{equation}\label{sign_scalare}
sig(x)^{\alpha}=sign(x)\vert x\vert^{\alpha}    
\end{equation}
where $\alpha>0$, $x \in \mathbb{R}$ and $sign(\cdot)$ is the signum function.
\end{definition}

%According to \cite{lu2013finite,bhat2000finite}, if $x=[x_{1}, x_{2}, \cdots, x_{n}]^{\top} \in \mathbb{R}^{n} $, %then
%\begin{equation}\label{sig_vettore}
%sig(x)^{\alpha}=sign(x) \odot \vert \widehat{x}\vert^{\alpha} 
%\end{equation}
%where $sign(x)=[sign(x_{1}), sign(x_{2}), \cdots, sign(x_{n})]^{\top}$ and $\vert \widehat{x}\vert^{\alpha}=[\vert x_{1}\vert^{\alpha}, \vert x_{2}\vert^{\alpha}, \cdots, \vert x_{n}\vert^{\alpha}]^{\top}$.
\noindent
Furthermore, for $\alpha>0$ and $x \in \mathbb{R} \setminus \{0\}$ the following properties~\cite{sun2016finite},

\begin{subequations}\label{inequality_sig}
\begin{align}
\frac{\partial sig(x)^{\alpha}}{\partial x}= \alpha \vert x \vert^{\alpha-1} \label{inequality_sig_1} \\
\frac{\partial \vert x\vert^{\alpha}}{\partial x}= \alpha sig(x)^{\alpha-1}, \label{inequality_sig_2}
\end{align}
\end{subequations}
hold. \\
\noindent
Finally, we recall the following finite time Lyapunov Theorem \cite{bhat2000finite}.

\begin{theorem}
\label{finite_time_theorem}
Consider the system $\dot{x}=f(x)$, where $x \in \mathbb{R}^{n}$, $f: U \rightarrow \mathbb{R}^{n}$ is a continuous function on an open neighborhood $U \subseteq \mathbb{R}^{n}$ of the origin and $f(0)=0$.
Suppose there exists a continuous positive definite $V(x): U \rightarrow \mathbb{R}$, a real number $c >0$ and $\alpha \in (0;1)$ and an open neighborhood $U_{0} \subset U$ of the origin such that $\dot{V}(x)+ c (V(x))^{\alpha} \leq 0$, $x \in U_{0} \setminus {0}$. Then $V(x)$ approaches to $0$ in finite time $T$ with 
\begin{equation}
\label{settling_time_function}
T \leq \frac{(V(x(0)))^{1-\alpha}}{c(1-\alpha)}.
\end{equation}
\end{theorem}

\section{Formulation of the Cooperative Intersection Crossing Problem}\label{sec:problemstatmenet}
%{\bf STEFANIA}
Consider $N$ vehicles approaching a generic traffic junction from $\mu$ different two-lane roads, with no traffic lights or any other kind of signalling provided by an infrastructure acting as central arbiter. All vehicles have to overpass the intersection while avoiding collisions and minimizing the crossing time (virtually, even with no need for a stop).
In Cooperative Intersection Crossing (CIC) problem, vehicles are also assumed to be connected via Vehicle-to-Vehicle (V2V) communication in order to share information about their own trajectory and their local state (e.g., see \cite{coelingh2012all} and references therein).
Hence, the practical implementation of a CIC strategy is heavily based on a reliable V2V communication network in the urban areas for guaranteeing the smooth and safe crossing of the vehicles through the intersection. 
Nowadays this is feasible by leveraging on-board modems, that make vehicles able to share their data across a urban cellular network (see \Cref{tab:wifivs5g} where the main cellular features are compared to ones of the vehicular Wi-fi network based on the IEEE 802.11p protocol).
%{\bf at Simon: la frase seguente proviene dal lavoro di Luca, ma nn è chiara\\
%Indeed, the information shared across the network is handled by a remote $5G$ endpoint which, virtually, can dynamically modify its structure depending from specific conditions. }
\\
% Given a generic intersection, its central polygonal zone is the Conflicting Area (CA), i.e., the core of the intersection where collisions could occur, while the larger circular zone around the CA, with radius $r_{cz}$, is the Cooperation Zone (CZ), i.e., the zone where vehicles interact (see  \cref{fig:incrociogenerico}). 
Given a generic intersection, we define its central polygonal zone as the \emph{Conflicting Area (CA)}, i.e., the part of the intersection where collisions could occur, while the larger circular zone around the CA, with radius $r_{cz}$, is referred to as the \emph{Cooperation Zone (CZ)}, i.e., the zone where vehicles interact (see \cref{fig:incrociogenerico}).
The objective of the CIC is that each vehicle in the CZ autonomously regulates its motion, cooperating with its neighbors, to occupy the CA in a mutually exclusive fashion, without side and rear-end collisions \cite{wu2015distributed}. Namely, at any time instant at most one vehicle is allowed to drive without stopping within the CA. Note that the traffic flow at the intersection may be continuously. However, for a specific time interval, we only need to consider a restricted group of $N$ vehicles that are approaching the junction \cite{li2006cooperative}. Under this assumption, as shown in Fig. \ref{fig:incrociogenerico}, vehicles inside the CZ will be considered as the group that currently takes part in cooperative crossing, whereas vehicles outside the CZ will be postponed to the next negotiation slot.\\
From a control perspective, it is assumed that the path-following is ensured by %uncooperative steering, 
a lower-level path follower, while the vehicles velocity and the safe spacing among vehicles is automatically achieved via a cooperative control based on the {\em virtual platoon} concept \cite{medina2017cooperative}.
In other words, the two-dimensional intersection problem is simplified into a one-dimensional {\em virtual} platoon control problem (as shown in Fig. \ref{fig:Ordinamento}). The crossing sequence is negotiated among the connected vehicles based on their actual distance from the intersection center, %that hence is transformed into the assignation of a platoon index and of a proper and safe spacing that have to be guaranteed by on-board cooperative controllers (e.g., see \cite{medina2017cooperative} and references therein for details about the geometrical transformation).
which is mapped into a crossing order, i.e. the closest vehicle goes first.
\begin{figure}[t]
\centering
\includegraphics[width=0.5\columnwidth]{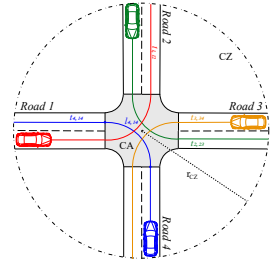}
\caption{A possible traffic junction scenario ($\mu$ = 4). Self-driving connected cars cooperate for crossing the Conflicting Area (CA). Once inside the Cooperative Zone (CZ), the vehicle $i$ may choose one of the possible trajectories $t_{i,pq}$ starting from the road $p$ where it is initially located.}
\vspace{-10pt}
\label{fig:incrociogenerico}
\end{figure}
\begin{figure}[t]
\centering
\subfloat[]{
%\vspace{-10pt}
\includegraphics[width=0.75\columnwidth]{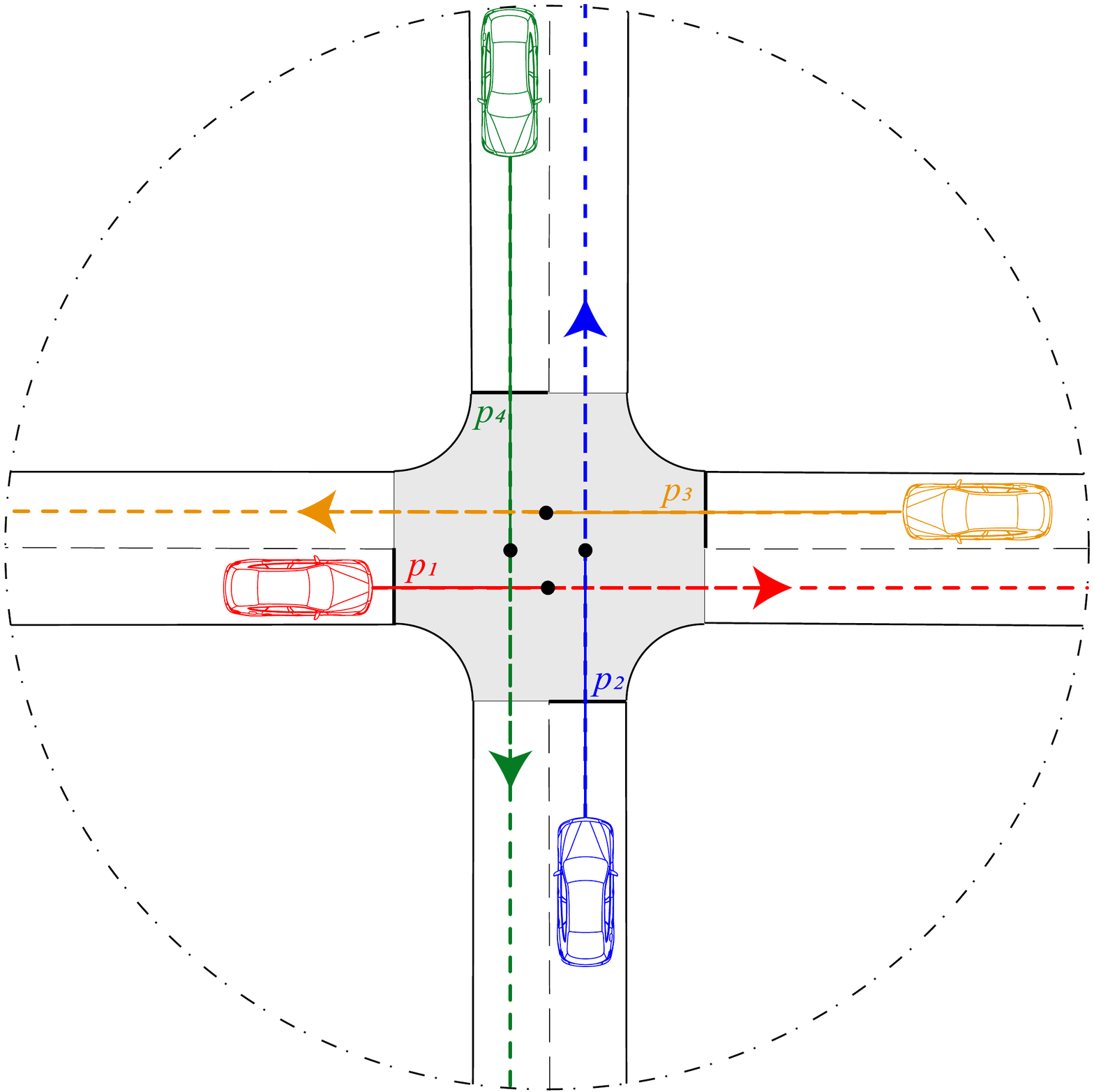}\label{fig:incrocioOrdineIngresso}}\\
\vspace{-10pt}
\subfloat[]{ 
\includegraphics[width=0.75\columnwidth]{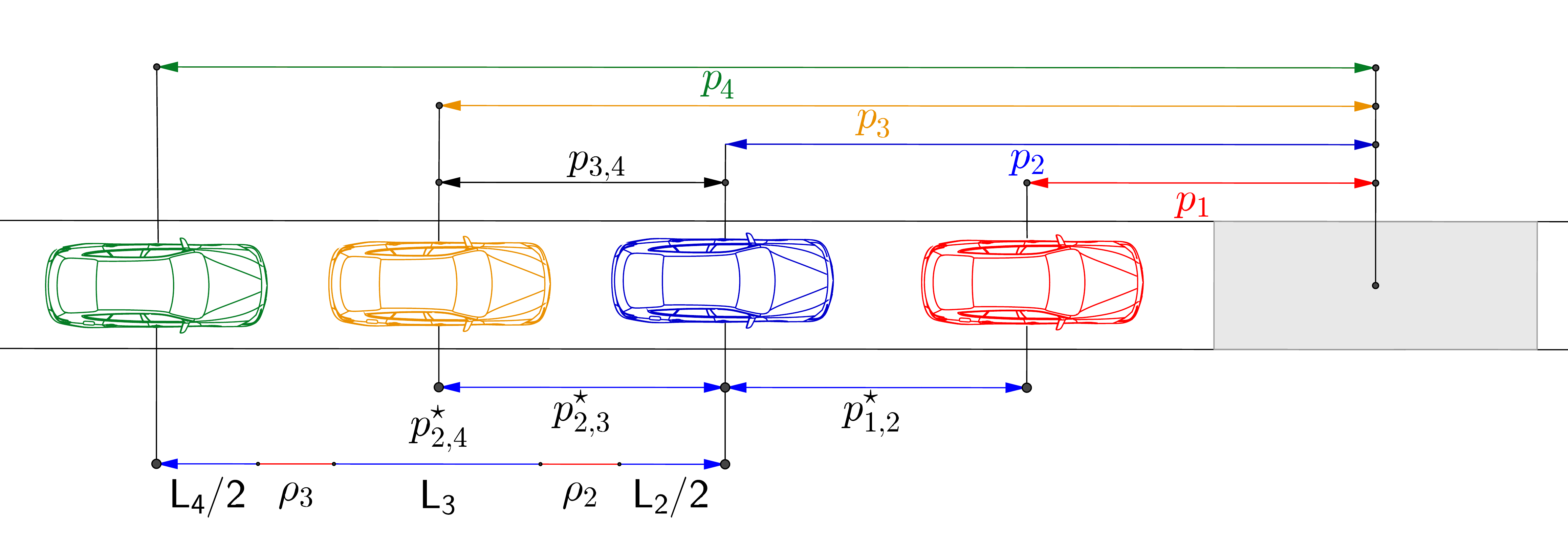}\label{fig:VirtualPlatoon}
}
%\vspace{10pt}
\caption{CIC. a): autonomous vehicles approaching the traffic junction  b): recast into a virtual platoon problem (on the base of the position from the centre $p_{i}(t)$.
}
\vspace{-15pt}
\label{fig:Ordinamento}
\end{figure}
%Hence, the core idea is to re-organize and control all vehicles driving within the CZ as a virtual platoon, ordered on the basis of the distance $p_i(t)$ of each vehicle from the center of its trajectory, $t_{i,qg}$ $\forall i \in \mathcal{V}_N$ . 
%Indeed, leveraging some coordinate transformations and mathematical manipulations, it is possible to map the vehicle coordinate vector $r_i(t)$, provided by the GPS receiver, into the position $p_i(t)$ along the trajectory and the relative positions $p_{ij}(t)$ between vehicles $i$ and $j$ ($\forall i \in \mathcal{V}_N$) (e.g., see \cite{medina2017cooperative} and references therein for details).\\
%The vehicles ordering into the platoon w.r.t $p_i(t)$ clearly corresponds to a crossing order, so that the closest vehicle crosses first. 
Since side and rear-end collision must be avoided, a desired spacing policy has to be imposed within the virtual formation, i.e., vehicles have to reach and maintain pre-fixed inter-vehicular gaps as they move with a common velocity. Specifically, the desired distances among virtual platoon members, say $p_{ij}^{\star}$ ($ \forall (i,j) \in \mathcal{E}_{N}$), have to be selected so to ensure that real vehicles access exclusively the CA, while the achievement of a common velocity guarantees that the desired formation will be preserved once reached. 
It is important to highlight that collisions are prevented only if the cooperative algorithms guarantee the achievement of the desired virtual formation in a prescribed finite time $T$ before the first vehicle enters into the CA.\\
%We assume that each self-driving vehicle $i$ ($i \in \mathcal{V}_{N}$) is able to measure its own position and speed exploiting on-board sensors, and hence to act accordingly on its acceleration/braking control systems for following, also thanks to on-board steering, its own trajectory $t_{i,qg}$ linking the road $q$, where the vehicle is initially located, with the road $g$, where the vehicle is heading to.
Now the CIC problem can be stated as follows. Let the following dynamic behavior for each vehicle within the CZ:
\begin{equation}
\begin{array}{cc}
\dot p_i(t)=v_i(t)\\
\dot v_i(t)=u_i(t), 
\end{array}
\label{eq:LongitudinalModel}
\end{equation}
being $p_i(t)$ the position of each vehicle $i$, expressed as its distance from the center of its trajectory $t_{i,qg}$ (linking the road $q$, where the $i$-th vehicle is initially located, with the road $g$, where the $i$-th vehicle is heading to, as shown in Fig. \ref{fig:Ordinamento}), and $v_i(t)$ its velocity. The cooperative control %goal can be given as:
problem can be formulated as follows:
\begin{problem}\label{problem}({\em CIC -- Cooperative Intersection Crossing -- in finite time}). Given the virtual platoon, obtained by organizing the $N$ vehicles within the CZ in ascending order of distances from the center of their trajectories $p_i(t)$ ($\forall i \in \mathcal{V}_N$), find a distributed cooperative control protocol $u_i(t)$ such that $\forall (i,j) \in \mathcal{E}_N$ the achievement of the following desired formation is guaranteed in a finite-time $T$:
\begin{equation}
\begin{array}{c}
\vert p_{i}-p_{j} \vert \rightarrow p_{ij}^{\star} \\%\label{constraint2}\\
\vert v_{i}-v_{j} \vert \rightarrow 0 \\ %\label{constraint1}
\end{array}
\end{equation}
being $p_{ij}^{\star}= r_{ij} + hv_i$ the safe virtual inter-vehicular gaps where $r_{ij}$ is the stand-still distance between the vehicle $i$ and the vehicle $j$, $h$ is the headway time, and $v_i$ is the velocity of the $i$-th vehicle.
%(calculated according to Algorithm \ref{MinimumAlgorithm}).
\end{problem}
    ~ %add desired spacing between images, e. g. ~, \quad, \qquad, \hfill etc. 
\section{Design of the Finite-Time Distributed Cooperative Control for CIC} \label{sec:ControlProtocol}
%{\bf STEFANIA}
In order to solve Problem \ref{problem},
here we design a distributed control strategy relying on communication with the neighbouring vehicles. The choice of a distributed approach sharply reduces  the  computational  load  on  the  remote endpoint and hence is more  efficient  from  a  computational  point  of  view. Also it can easily and quickly scale to an increasing number of  vehicles  approaching  the  intersection.
The distributed nonlinear control law for each vehicle $i$ is given as:
\begin{equation}\label{control}
\begin{array}{l}
u_{i}(t)= -\sum\limits_{j=1}^{N} a_{ij} sig(p_{i}(t)-p_{j}(t)-p_{ij}^{\star})^{\frac{2\alpha}{1+\alpha}}\\
\qquad \qquad - \sum\limits_{j=1}^{N} a_{ij} sig(v_{i}(t)-v_{j}(t))^{\alpha},
\end{array}
\end{equation}
where $\alpha \in (0;1)$ and $sig(\cdot)$ is defined as in \cref{sig}. Moreover, $a_{ij}$ models the topology of the underlying connected communication graph $\mathcal{G}_N$, i.e., the presence/absence of a communication link between the $i$-th and $j$-th vehicle ($a_{ij} = 0 \; \forall j \notin \mathcal{N}_{i}$ as reported in \cref{preliminaries}).
Note that the controller is distributed in the sense that each agent requires only relative position and velocity measurements of its neighboring agents.
\subsection{Finite-Time Stability Analysis of the Closed-loop network}
Given (\ref{eq:LongitudinalModel}) and (\ref{control}), the closed-loop dynamics for the $i$-th vehicle can be derived ($\forall i \in \mathcal{V}_N$) as:
\begin{subequations}\label{closed-loop_i}
\begin{align}
\dot{p}_{i}(t)&=v_{i}(t)\label{closed-loop_i1} \\
\dot{v}_{i}(t)&= -\sum\limits_{j=1}^{N} a_{ij} sig(e_{ij}(t))^{\frac{2\alpha}{1+\alpha}} - \sum\limits_{j=1}^{N} a_{ij} sig(v_{i}(t)-v_{j}(t))^{\alpha},\label{closed-loop_i2}
\end{align}   
\end{subequations}
where $e_{ij}(t)=p_{i}(t)-p_{j}(t)-p_{ij}^{\star}$ is the distance error between vehicle $i$ and vehicle $j$ according to the desired spacing $p_{ij}^{\star}$.\\
We next establish a finite-time stability result. %which is fundamental to guarantee the safety  of the proposed distributed control low.  
\begin{theorem}
Consider $N$ self-driving vehicles, sharing information via V2V communication, with closed-loop longitudinal dynamics as in (\ref{closed-loop_i}). If the corresponding communication graph $\mathcal{G}_{N}$ is connected in the CZ, then the control strategy $u_{i}(t)$ in (\ref{control}) solves Problem \ref{problem}, i.e., it ensures that vehicles converge to the desired distance with a common velocity in a finite time $T$.
\end{theorem}
\begin{proof}
In order to solve our specific crossing problem, we propose the following Lyapunov function candidate
\begin{equation}\label{candidata}
\begin{array}{c}
 V(e_{ij}(t),v_i(t))=\sum\limits_{i=1}^{N} V_{i}
 %\sum\limits_{j=1}^{N} \int_{0}^{e_{ij}(t)} a_{ij} sig(s)^{\frac{2\alpha}{1+\alpha}} {\rm ds} + \frac{1}{2} \sum\limits_{i=1}^{N} v_{i}^{2}(t),
\end{array}
 %V(e_{ij}(t),v_i(t))=\sum\limits_{i=1}^{N} V_{i} %\sum\limits_{j=1}^{N} \int_{0}^{e_{ij}(t)} a_{ij} sig(s)^{\frac{2\alpha}{1+\alpha}} {\rm ds} + \frac{1}{2} \sum\limits_{i=1}^{N} v_{i}^{2}(t),
\end{equation}
where
$$
V_{i}=\sum\limits_{j=1}^{N} \int_{0}^{e_{ij}(t)} a_{ij} sig(s)^{\frac{2\alpha}{1+\alpha}} {\rm ds} + \frac{1}{2} v_{i}^{2}(t),
$$
which is positive definite, w.r.t. $e_{ij}(t)$ and $v_{i}(t)$ $\forall i,j=1, \cdots,N$, $i \neq j$. Note that this can be easily shown leveraging properties (\ref{inequality_sig}).\\
Differentiating the Lyapunov function along the trajectories $p_i(t)$ and $v_i(t)$, solutions of system (\ref{closed-loop_i}), it follows
\begin{equation}\label{derivata_candidata_zero}
\begin{array}{l}
\dot{V}(e_{ij}(t),v_i(t)) = \sum\limits_{i=1}^{N} \sum\limits_{j=1}^{N} a_{ij} sig(e_{ij}(t))^{\frac{2\alpha}{1+\alpha}} \dot{p}_{i}(t)\\
\qquad \qquad \qquad 
+\sum\limits_{i=1}^{N} v_{i}(t) \dot{v}_{i}(t), 
\end{array}
\end{equation}
and from (\ref{closed-loop_i})
\begin{equation}\label{derivata_candidata}
\begin{array}{l}
\dot{V}(e_{ij}(t),v_i(t)) = \sum\limits_{i=1}^{N} \sum\limits_{j=1}^{N} a_{ij} sig(e_{ij}(t))^{\frac{2\alpha}{1+\alpha}} v_{i}(t)\\
+\sum\limits_{i=1}^{N} v_{i}(t) u_{i}(t)
= \sum\limits_{i=1}^{N} \sum\limits_{j=1}^{N} a_{ij} sig(e_{ij}(t))^{\frac{2\alpha}{1+\alpha}} v_{i}(t)+\\
\sum\limits_{i=1}^{N} v_{i}(t) \Big(\sum\limits_{j=1}^{N} a_{ij} sig(e_{ij}(t))^{\frac{2\alpha}{1+\alpha}} 
- \sum\limits_{j=1}^{N} a_{ij} sig(v_{i}(t)-v_{j}(t))^{\alpha} \Big) \\
= -\sum\limits_{i=1}^{N} \sum\limits_{j=1}^{N} v_{i}(t) a_{ij} sig(v_{i}(t)-v_{j}(t))^{\alpha}.
\end{array}
\end{equation}
Since $sig(\cdot)$ is an odd function, while the adjacency matrix A (defined in section \ref{preliminaries}) is symmetric under the assumption of connected undirect graph $\mathcal{G}_{N}$, it follows that (\ref{derivata_candidata}) can be recast as
\begin{equation}\label{derivata_candidata_recast}
\begin{array}{l}
\dot{V}(e_{ij}(t),v_i(t)) = \sum\limits_{i=1}^{N} \sum\limits_{j=1}^{N} v_{i}(t) a_{ij} sig(v_{j}(t)-v_{i}(t))^{\alpha}\\
= \frac{1}{2}\sum\limits_{i=1}^{N} \sum\limits_{j=1}^{N}v_{i}(t)a_{ij}sig(v_{j}(t)-v_{i}(t))^{\alpha}\\ 
+ \frac{1}{2}\sum\limits_{i=1}^{N} \sum\limits_{j=1}^{N} v_{j}(t) a_{ij} sig(v_{i}(t)-v_{j}(t))^{\alpha}\\
= \frac{1}{2}\sum\limits_{i=1}^{N} \sum\limits_{j=1}^{N} (v_{i}(t)-v_{j}(t)) a_{ij} sig(v_{j}(t)-v_{i}(t))^{\alpha}\\
= -\frac{1}{2}\sum\limits_{i=1}^{N} \sum\limits_{j=1}^{N} (v_{i}(t)-v_{j}(t)) a_{ij} sig(v_{i}(t)-v_{j}(t))^{\alpha}\\
 = -\frac{1}{2}\sum\limits_{i=1}^{N} \sum\limits_{j=1}^{N} a_{ij}\vert v_{i}(t)-v_{j}(t) \vert^{1+\alpha}.
\end{array}
\end{equation}
Let now introduce, for sake of brevity, a more compact notation for the distance errors by indicating each couple of indices $({i,j}) \in \mathcal{E}_{N}$ with a new index $\rho$. In so doing, errors are referred as elements of the following set $e_\rho(t) \in \{e_{ij}(t): i,j=1,\dots,N; i \neq j \}$ for $\rho =1, \dots, m$, being $m=\vert \mathcal{E}_{N} \vert$, i.e., being $m$ equal to the cardinality of the edge set (according to the nomenclature in \cref{preliminaries}).\\
Now it is possible to define the following distance error vector as $e(t)=[e_{1}(t), e_{2}(t), \cdots, e_{m}(t)]^\top$, while the velocity vector is $v(t)=[v_{1}(t), v_{2}(t), \cdots, v_{N}(t)]^\top$.\\
Leveraging the above notation, from (\ref{derivata_candidata_recast}) one has that $\dot{V}(e(t),v(t)) \leq 0$ and, hence, that $V(e(t),v(t))\leq V(e(0),v(0))=V_{0}$, which indicates that $e(t)$ and $v(t)$ are bounded $\forall t \geq 0$.
In addition, since the Lyapunov function $V(e(t),v(t))$ is radially unbounded \cite{slotine1991applied} (see its structure in (\ref{candidata})) it follows that the invariant set $\Omega$, defined as
\begin{equation}\label{compattoV0_omega}
\Omega=\{ e(t) \in \mathbb{R}^{m}, \; v(t) \in \mathbb{R}^{N} \; : V(e(t),v(t))\leq V_{0} \},
\end{equation}
is compact.
Thus, from the LaSalle Invariance Principle \cite{slotine1991applied} one has that all trajectories that start from $\Omega$ converge to the largest invariant set defined as
\begin{equation}\label{compattoV0_largest_invariant_set}
S=\{ e(t) \in \mathbb{R}^{m}, \; v(t) \in \mathbb{R}^{N} \; : \dot{V}(e(t),v(t))=0 \}.
\end{equation}
Note that, since the underlying undirected communication graph is connected, $\dot{V}(e(t),v(t))=0 $ implies that all vehicles velocities approach the average velocity ($i,j=1, \dots,N \;, \forall j \neq i$)
$$v_{i}(t)=v_{j}(t)=v^{\star} = \sum_{i \in \mathcal{V}_N } \frac{v_i}{N},$$ which in turn implies that at steady state $u_{i}(t)=u_{j}(t)=0$. \\
From (\ref{control}),
\begin{equation}\label{control2}
u_{i}(t)=-\sum\limits_{i=1}^{N}\sum\limits_{j=1}^{N} a_{ij} sig(p_{i}(t)-p_{j}(t)-p_{ij}^{\star})^{\frac{2\alpha}{1+\alpha}}=0
\end{equation}
implies that $sig(p_{i}(t)-p_{j}(t)-p_{ij}^{\star})^{\frac{2\alpha}{1+\alpha}}=0$, or equivalently that $p_{i}(t)-p_{j}(t)=p_{ij}^{\star}$. %($\forall i,j=1,\cdots,N,\; i\neq j$). 
In so doing, it is proven that all vehicles asymptotically converge to the fixed desired formation configuration.\\
In the following, we will prove that the convergence of the velocity alignment, as well as the
convergence of formation stabilization, is achieved in finite time. To this aim, we leverage the homogeneity property of the Lyapunov function according to \cite{lu2013finite,bhat2000finite,sun2016finite,bhat1997finite}.\\
Given (\ref{candidata}) and (\ref{derivata_candidata_recast}), for any $\mu>0$ there holds
\begin{equation}\label{homogeneity1}
V(\mu^{\frac{\alpha+1}{\alpha}}e,\mu v)= \mu^{2}V(e,v),
\end{equation}
\begin{equation}\label{homogeneity2}
\dot{V}(\mu^{\frac{\alpha+1}{\alpha}}e,\mu v)= \mu^{1+\alpha} \dot{V}(e,v),
\end{equation}
which verifies the homogeneity properties of $V(e,v)$ and  $\dot{V}(e,v)$. Note that for the sake of simplicity, the time dependence has been omitted.\\
From (\ref{homogeneity2}), with $\mu=[V(e,v)]^{-\frac{1}{2}}$ we have
\begin{equation}\label{maggiorante}
\begin{split}
\frac{\dot{V}(e,v)}{V(e,v)^{\frac{1+\alpha}{2}}}&=\dot{V}(V(e,v)^{-\frac{\alpha+1}{2\alpha}}e,V(e,v)^{-\frac{1}{2}} v)\\
& \leq \max_{(e,v) \in \Upsilon}\dot{V}(e,v)
\end{split}
\end{equation}
where 
\begin{equation}\label{Compatto2}
\begin{split}
\Upsilon=& \Big \lbrace  e \in \mathbb{R}^{m}, \; v \in \mathbb{R}^{N} \setminus \{ (0^T,0^T)^T\}: \\
 & V(e,v)= V\Big(V(e,v)^{-\frac{\alpha+1}{2\alpha}}e,V(e,v)^{-\frac{1}{2}} v \Big) \Big\rbrace.
\end{split}
\end{equation}
From homogeneity property in (\ref{homogeneity1}), it follows
\begin{equation}\label{homogeinity_applied}
\begin{split}
V\Big(V(e,v)^{-\frac{\alpha+1}{2\alpha}}e,V(e,v)^{-\frac{1}{2}} v \Big)= \Big(V(e,v)^{-\frac{1}{2}}\Big)^2 V(e,v)=1.
\end{split}
\end{equation}
Therefore, $\Upsilon= \{e \in \mathbb{R}^{m}, \; v \in \mathbb{R}^{N} : V(e,v)=1 \}$ is a compact set due to the radially unbounded property of $V(e,v)$.
Since $\dot V (e,v)$ is continuous and non-positive on the compact set $\Upsilon$, we have 
\begin{equation}\label{maggiorante2}
\frac{\dot{V}(e,v)}{V(e,v)^{\frac{1+\alpha}{2}}} \leq \max_{(e,v) \in \Upsilon}\dot{V}(e,v)= -c 
\end{equation}
where $c \geq 0$.
Furthermore, by the fact that 
\begin{equation}\label{compattoV0_final}
\{ e(t) \in \mathbb{R}^{m}, \; v(t) \in \mathbb{R}^{N} \; : \dot{V}(e(t),v(t))=0 \}= \{(0^{\top},0^{\top})^{\top} \},
\end{equation}
one obtains $c>0$.
Therefore, condition (\ref{maggiorante2}) implies that 
\begin{equation}\label{maggiorante3}
\dot{V}(e,v) \leq -c V(e,v)^{\frac{1+\alpha}{2}}.
\end{equation}
Since $\frac{1+\alpha}{2} \in (0;1)$, from Theorem \ref{finite_time_theorem} it follows that the closed-loop system is finite time stable with settling time $T$ such that
\begin{equation}\label{estimate}
    T \leq \frac{2}{c(1-\alpha)}V(e(0),v(0))^{\frac{1-\alpha}{2}}.
\end{equation}
This completes the proof.
\end{proof}
\begin{remark}
By constructing a Lyapunov function for the closed-loop system, the settling time is estimated by computing the Lyapunov function value at the initial point according to \cite{zhao2016distributed}. 
\end{remark}
\begin{remark}
According to (\ref{estimate}), it is possible to tune the control gain $\alpha$ to select a proper upper-bound for the convergence time.  
\end{remark}
    ~ %add desired spacing between images, e. g. ~, \quad, \qquad, \hfill etc. 

\section{Communication Software: the Hermes module}
\label{sec:communicationsoftware}
%{\bf{Questa section andrebbe prima della parte di controllo poichè stiamo puntando a TON}}
%{\bf{By spromano: Mmhhh, secondo me sta bene qui, visto il flusso logico del paper...}}

%In this section we discuss the design and implementation of a generalised, real-time, reliable and low-latency message exchanging system that we called \textit{Hermes} and which acts as support for the control algorithm described in Section~\ref{sec:ControlProtocol}. 
\LMC{In this section we discuss the design and implementation of a generalised, real-time, low-latency and reliable message exchanging system that we called \textit{Hermes}. Hermes acts as an application-level communication infrastructure to support the control algorithm discussed in Section~\ref{sec:ControlProtocol}. } Thanks to the level of abstraction that has been used in the architecture design, \textit{Hermes} can provide a generic road user with traffic information, withstanding the differences intrinsic in the specific configuration of each vehicle.\\
Moreover, even though it is reasonable to assume that each involved vehicle is able to take autonomous decisions based on the information received from the traffic controller, the latter should also provide recipients with additional information (also called \textit{control-side information} in this context) that can bias the final control decision (i.e., the one initially taken locally on the vehicle). An example scenario where such control-side information comes in handy, could be the need to prioritize a vehicle. In fact, the traffic controller might be willing to force the order of vehicles that are about to cross an intersection because a high priority vehicle (e.g., an ambulance) is approaching. \LMC{To this extent, the communication software is highly decoupled from the adopted control strategy. 
In order to be reliable, the system must be aware of the status of the connection it has established with any user. In this sense, it is possible to identify two main classes of listeners: 
\begin{enumerate*}
    \item \emph{vehicles}, either autonomous or human-driven, for which reliability must be guaranteed;
    \item  \emph{monitors}, with no specific reliability requirements.
\end{enumerate*}}
%0, as shown in Table~\ref{tab:roadusers}. 
Indeed, while vehicles are supposed to proactively leverage the data they get from the traffic controller, monitors are just passively listening to control data in order to, e.g., assess the performance of the overall system. They hence demand for less stringent requirements in terms of reliability.

%\begin{table}
%\centering
%\begin{tabular}{|l l l|} 
%\hline
%\textbf{User} & \textbf{Example} & \textbf{Reliability Level}\\
%\hline
%Autonomous & Autonomous vehicles. & \textbf{Guaranteed} \\ 
%Non Autonomous & Human driven vehicles & \textbf{Guaranteed} \\ 
%Monitors &  Monitor interface & \textbf{Not Guaranteed} \\
%\hline
%\end{tabular}
%\caption{User Classification} 
%\label{tab:roadusers}
%   \vspace{-10pt}
%\end{table}

The entire communication system has been designed to be easily deployable and highly scalable, so to seamlessly cope with an increasing number of vehicles. As to the \textit{traffic manager}, it has been developed using state of the art software and well-known programming methodologies. As it will come out from the next sections, it is easily extendable and maintainable. 

%In the simplest scenario, vehicles are equipped with the following components:
\LMC{Packets exchanged by the communication software have been shaped to cope with data sensed by the following instrumentation, deployed on our test vehicles:}
\begin{enumerate*}[label=\alph*]
    \item High precision \emph{GNSS system}; %The GNSS signal can be adjusted with additional information retrieved both from Radio and LTE signals, depending on the area coverage.
    \item \emph{pre-5G LTE client modem} to exchange messages with other vehicles through a $5$G network;
    \item \emph{Proxy} acting as interface between the vehicle integrated hardware and the \textit{Hermes} module; %CAN: An element acting as a proxy to the CAN (Controller Area Network) interface to the network. Through the proxy it is possible to sense information from the car integrated hardware (such as velocity, steering angle, etc.), as well as interact with car integrated actuators.
    \item \emph{Inertial Measurement Unit (IMU)} providing information such as acceleration of the vehicle.
    %\item On-board computer, receiving information through an ethernet interface, and calculating the control output according to a pre-defined policy. The above mentioned instrumentation is inter-connected through an on-board LAN.
\end{enumerate*}

The configuration of the system is summarized in Figure \ref{fig:pp_block_diagram}.  %\ref{fig:sad_overall} 

%\begin{figure}[!ht]
%    \centering
%    \includegraphics[width=.8\columnwidth]{imgs/planning2.png}
%    \caption{Overall system architecture}
%    \label{fig:sad_overall}
%    \vspace{-15pt}
%\end{figure}

\subsection{System domain model}

Hermes is able to manage the traffic among several road users. In order to do that, the system monitors the status of mobile nodes within a desired area, constantly receiving status messages and \LMC{providing clients with network updates}. %sharing the overall status of the vehicle network over $5$G. 
%In this sense, the high level overall architecture is fairly easy to describe.% and it is showed in Figure \ref{fig:cdhlsm}. 
The core element of the system is the \textit{Mobile Node}, a virtual representation of a subscribed road user. Each \textit{Mobile Node} element, which is univocally identified by a combination of \textit{id} and \textit{name}, stores status information such as \textit{position, speed} and \textit{acceleration} of the corresponding physical node in the \textit{Status} fields. Eventually, the \textit{Timestamp} field is updated with the time value of the last received message. This proves useful when detecting possible communication and vehicular anomalies. Hence, the traffic manager stores the virtual representation of each subscribed node in a subscription list, along with control information such as a \textit{Sequence Number} and a \textit{Global Timestamp}. 

%\begin{figure}[h!]
%    \centering
%    \includegraphics[width=.8\columnwidth]{imgs/HighLevel_SystemDomainModel.PNG}
%    \caption{Class Diagram - High Level System Domain Model}
%    \label{fig:cdhlsm}
%    \vspace{-15pt}
%\end{figure}

\subsection{System dynamics} \label{sec:entities}
%A dynamic representation of the behaviour of the system is illustrated in Figures \ref{fig:sdhlsub} and \ref{fig:sdhlsu}. Firstly, 
When a node is activated, it is requested to send a subscription request to the traffic manager in order to have its status tracked over the network. Once the subscription has been submitted and approved by the manager, the node in question can start sharing status updates by sending information such as current position, speed and acceleration, along with the above mentioned control data (such as a timestamp and a local sequence number). The traffic manager, in turn, periodically %\footnote{In this work a frequency of 20Hz has been chosen} 
sends the latest updates via multicast to connected nodes, making each of them aware of the status of all vehicles passing through the covered area. Thus, from a high-level point of view, the traffic manager acts as a mean to share `all to all' connection information among road users.\\
It is important to highlight a difference between \emph{local} and \emph{global} sequence number. The former is calculated locally to each car and belongs to the vehicle status data structure. It is used to distinguish between two packets originating from the same mobile node. The latter originates on the server and it is associated with a \textit{traffic update} packet sent by the traffic manager. Two packets with different global sequence number have both origin in the traffic manager server and are born in two different moments. In particular, the packet with the highest sequence number is the newest one. 

%\begin{figure}[h!]
%    \centering
%    \includegraphics[width=.9\columnwidth]{imgs/HighLevel_Subscription.PNG}
%    \caption{Sequence Diagram - High Level Subscription}
%    \label{fig:sdhlsub}
%    \vspace{-15pt}
%\end{figure}
%
%\begin{figure}[h!]
%    \centering
%    \includegraphics[width=.9\columnwidth]{imgs/HighLevel_StatusUpdate.PNG}
%    \caption{Sequence Diagram - High Level Status Update}
%    \label{fig:sdhlsu}
%        \vspace{-10pt}
%\end{figure}

In order to exchange their own status, road users leverage a dedicated class of messages, %as sketched in Figure~\ref{fig:to5GPoC}. Each message falling under this category is 
referred to as a \textit{to5GPoC}  message in this work. 

%\begin{figure}[h!]
%    \centering
%    \includegraphics[width=0.5\columnwidth]{imgs/FromVehicleMessage.PNG}
%    \caption{Message Description - to 5G PoC}
%    \label{fig:to5GPoC}
%        \vspace{-15pt}
%\end{figure}

Since the designed architecture envisages the presence of different types of vehicles, each such vehicle firstly has to declare its own nature (type of vehicle and type of aid in guide) by properly filling in the \textit{Vehicle Type} field of a \textit{to5GPoC} message. Also, this architecture uses the field \textit{Vehicle Name} to distinguish between two different vehicles. In a real scenario, this field has been thought to be filled in with unique information such as plate number or VIN (Vehicle Identification Number). The following data are also included in this kind of message: a) \textit{GNSS Coordinates}: current Position of the vehicle according to the WGS84 standard \cite{wgs84}; b) \textit{GNSS Heading}: heading angle of the running vehicle; c) \textit{Vehicle Speed}: current Speed of the vehicle. This field is composed of two values, namely the \textit{latitudinal speed} and the \textit{longitudinal speed}; d) \textit{Proximity}: here the vehicle declares its distance in meters from the intersection it is approaching and that is supposed to be autonomously negotiated. If the traffic manager has to handle more than one intersection, it can mix this information with data coming from the GNSS to figure out which specific intersection the vehicle is approaching; e) \textit{Connection Status}: this field is used to keep track of the connection status of the vehicle. The assumed value should be binary: \textit{Active} or \textit{Inactive}; f) \textit{Latency}: the latency measured across the connection between the node and the traffic manager, based on the last message received; g) \textit{Local Timestamp}: time at the mobile node when information encapsulation has taken place at application level; h) \textit{Local Sequence}: sequence number of the packet. This counter is managed by the local node and will be reset to 0 when the connection is restarted.\\
On the other side of the communication, the traffic manager is listening for updates from subscribed nodes, keeping its vehicle list up-to-date with received information. It then periodically shares the stored traffic data with road users. Messages sent from the server are classified as \textit{Traffic Update} messages and contain the following data:%and their structure is illustrated in Figure~\ref{fig:msgtrafficupdate}. 

%\begin{figure}[h!]
%    \centering
%    \includegraphics[width=.9\columnwidth]{imgs/FromTrafficManager.PNG}
%    \caption{Message Description - Traffic Update}
%    \label{fig:msgtrafficupdate}
%        \vspace{-10pt}
%\end{figure}
%
%The following data are shared from the Hermes Manager through the vehicular network:

\begin{enumerate*}[label=\alph*]
    \item \textit{Connected Nodes}: the number of mobile nodes currently subscribed to the network.
    \item \textit{Global Sequence Number}: the sequence number of the current message. It is computed at the traffic manager and keeps track of the messages sent. This counter only resets when the server is shut down.
    \item \textit{Global Timestamp}: time value extracted on the server machine when a \textit{Traffic Update} packet is encapsulated at application level.
    \item \textit{Control Side Information}: additional control data that may prove useful when there is a need to bias local decisions. This field can be used, for instance, if the manager wants one vehicle (i.e., ambulance) to be prioritized against the others, as well as in the case that stakeholders decide to switch to a more centralized control strategy.    
\end{enumerate*}

\subsection{Hermes Architecture}

\textit{Hermes} has been built by combining the simplicity of the \textit{Client-Server} pattern with the efficiency of the \textit{Multicast} communication paradigm. As already mentioned, the traffic manager plays a role which is of paramount importance in the overall architecture. It acts as a server, with the mobile nodes representing the clients. The \textit{Traffic Manager} server receives messages from clients, %makes due elaborations 
\LMC{elaborates them} and eventually shares the results among vehicles through a multicast session. %, as summarized in Figure \ref{fig:hermesclientserver}. 

\begin{figure}[h!]
    \centering
    \includegraphics[width=0.9\columnwidth]{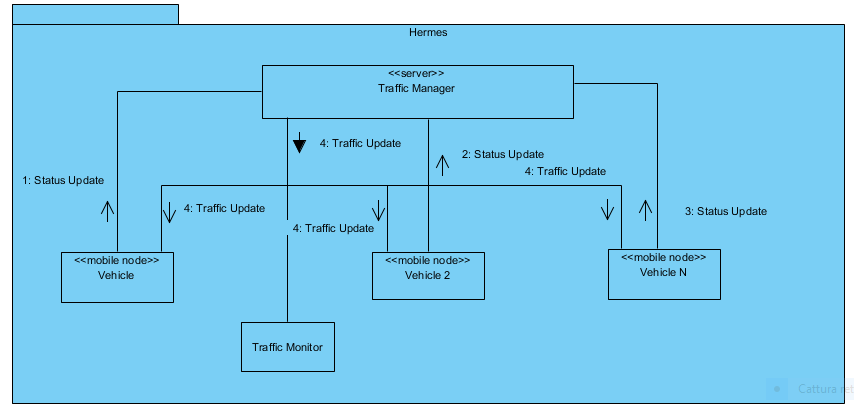}
    \caption{Hermes High Level Architecture}
    \label{fig:hermesclientserver}
        \vspace{-10pt}
\end{figure}

In order to add an additional level of reliability to the communication, messages between clients and traffic manager are exchanged over TCP rather than UDP.  With this choice we actually traded slightly decreased network responsiveness for improved communication reliability and this is justified by the critical nature of the application. In fact, even though in $5$G networks the reliability of the communication is guaranteed, at the physical layer, by URLLC~\cite{popovski2017ultra}, the further layer of reliability added by TCP does make it possible to safely use the \textit{Hermes Traffic Manager} software even in areas where $5$G coverage is not ensured and the connections are downgraded to standard LTE. It can also be noticed that, whenever $5$G coverage is already available, the overhead introduced by TCP is minimal~\cite{xylomenos1999tcp} (since a fault will cause a re-transmission at the physical layer and stay transparent to TCP) and it can be considered a fairly low price to be paid, which allows to gain portability toward classical LTE networks.

A fully-fledged version of the \textit{Hermes Traffic Manager} (that we called \emph{HermesJS}) has been implemented to carry out the experiments at AstaZero proving ground. The implementation of the service uses WebSockets~\cite{ws_ietf} as a means to exchange messages among involved entities. If the communication happens over a public network, the connection can be easily upgraded to secure WebSockets. In a nutshell, WebSockets represent an advanced technology that makes it possible to open an interactive, event-driven communication session between client and server with no need for polling to receive a reply. 
%On the server side, the \textit{socket.io} implementation of WebSockets has been used, within the context of a \textit{Node.JS} environment, which is known to be a highly secure and reliable one~\cite{nodebook}. The HermesJS server waits for incoming HTTP connections and upgrades them to the WebSocket protocol if they are supposed to interact with the traffic manager. 
\LMC{On the server side, the \textit{socket.io} implementation of WebSockets has been used, within the context of a \textit{Node.JS} environment. The HermesJS server waits for incoming HTTP connections and upgrades them to the WebSocket protocol if they are supposed to interact with the traffic manager. This choice comes as a result of a trade-off between availability, reliability and easy prototyping. While the reliability of \textit{Node.JS} is not proved, there are several studies (such as~\cite{nodebook}) on its availability and security attributes. However, it has proved to be reliable during our trials.}
The deployed proof of concept slightly diverges from the discussed design, particularly in relation to the way multicast is implemented. Indeed, according to the standard patterns, two different connections should be used by each client to send the status and receive road-traffic information. In this sense, outgoing information should travel towards the traffic manager across a dedicated \textit{client-server} route, while incoming data are supposed to be dispatched via multicast at network level. What actually happens in the discussed implementation is that the multicast paradigm is implemented at the application level. Communication between a node and the traffic manager actually happens, for a single TCP flow, via bidirectional unicast. Consequently, from a network perspective each vehicle opens a TCP connection towards the traffic manager and uses this stream both to send and receive messages. %as showed in Figures \ref{fig:lowlevelsd1}, \ref{fig:lowlevelsd2} and \ref{fig:lowlevelsd3}. The reported diagrams show, in detail, the exchanges taking place between clients and server. 
In this scenario, each client negotiates a session with the \textit{Traffic Manager}, by setting up a websocket connection towards it. 
%\begin{figure}[ht!]
%    \centering
%    \includegraphics[width=0.4\textwidth]{imgs/lowlevelsd1.PNG}
%    \caption{Subscribing to the Traffic Manager}
%    \label{fig:lowlevelsd1}
%       \vspace{-10pt}
%\end{figure}
Once the connection has been established, the node engages in a subscription operation. During this phase, it sends a subscription message along with an identifier. If such an identifier has not been taken yet, the traffic manager notifies the occurred subscription by sending a positive acknowledgement. Otherwise, the client will be disconnected. %(see Figure \ref{list:subscription}). 
After a successfully completed subscription, mobile nodes are able to communicate updates about their status.
% (see Fig.~\ref{fig:lowlevelsd2}). %The operation of status update is shortcut in \ref{list:statusupdate}. 
%\begin{figure}[ht!]
%    \centering
%    \includegraphics[width=0.4\textwidth]{imgs/lowlevelsd2.PNG}
%    \caption{Sending status updates to the Traffic Manager}
%    \label{fig:lowlevelsd2}
%       \vspace{-10pt}
%\end{figure}
In parallel, the \textit{Traffic Manager} broadcasts received information, along with optionally computed control side data, at a frequency of 20 Hz. The $20$ Hz update frequency has been chosen to strike a balance between the need for minimizing network traffic overhead on one side and that of maximizing the effectiveness of the communication on the other.%, further details will be given during the client discussion.
%(Fig.~\ref{fig:lowlevelsd3}). % as listed in \ref{list:trafficupdate}. 

%\begin{figure}[ht!]
%    \centering
%    \includegraphics[width=0.4\textwidth]{imgs/lowlevelsd3.PNG}
%    \caption{Clients getting traffic updates via push notifications}
%    \label{fig:lowlevelsd3}
%       \vspace{-10pt}
%\end{figure}
\subsection{Mobile Nodes}
In  order to enable test cars to communicate over $5$G, an additional module has been designed and developed, in accordance with the mobile node specifications. This software is highly asynchronous and is written in low level C++ code in order to allow for maximum performance. It is logically divided in two main components running in parallel in different threads:
a) \emph{Remote Sender}: a time-triggered asynchronous thread that periodically\footnote{A frequency of 20Hz has been used in the experiments} sends data about current state of the vehicle towards the $5$G PoC;
b) \emph{Remote Receiver}: an asynchronous thread triggered by an incoming message.  The main purpose of this method is the extraction of traffic data from the websocket data format and the initialization of an internal traffic data structure coherent with the car software and components.

\section{Test cars software architecture and experimental setup}
\label{sec:experimentalsetup}
%{\bf @SIMON: Credo che questa parte vada sopstata nella section a seguire e integrata in essa opoortunamente}

%The experimental setup was composed of three vehicles: two autonomous cars and a human-driven truck. Involved cars are a \textit{Volvo Car XC90} and a \textit{Volvo Car S90}, while the truck is a \textit{Volvo Truck FH16}. The three vehicles differ in their configuration as they come from different research centres. In particular, the XC90 Car and the truck have been kindly loaned to us, for this research work, by the Revere Laboratory at Chalmers. They have been equipped with an open-source driving system, also known under the name of \textit{OpenDLV}~\cite{berger2016open} \cite{berger2017containerized} \cite{benderius2018best}, developed by Christian Berger et al. at Chalmers. On the other hand, the test car S90 is the property of the AstaZero research centre and is able to drive itself through an \emph{ADB Pedal Robot}. The system in question is basically a robot able to control the steering and the pedals of the vehicle given a well-defined speed profile. 

The experimental %trial involves 
setup consists of three vehicles (namely two cars, Volvo XC90 and Volvo S90, and one Truck, Volvo FH16) that exchange information via the pre-$5G$ communication test network provided by Ericsson. Vehicles are heterogeneous in their masses, power-trains and on-board systems. Namely, the Volvo Car XC90 and the Volvo Truck FH16 %, provided by the Revere Laboratory of Chalmers University of Technology, 
are equipped with the open-source driving system OpenDLV \cite{berger2016open,berger2017containerized,benderius2018best}, while the Volvo Car S90, provided by the proving ground AstaZero, is equipped with an ADB Pedal Robot \cite{abd} that controls the longitudinal vehicle motion by acting on its throttle/brake pedals. The robot can be controlled through a proprietary interface that, in this context, has been accessed with the \textit{Matlab Realtime} tool.
In the following we detail the main on-board hardware devices and software components.
% The cooperative driving system for each vehicle has been configured to allow longitudinal actuation and access to on-board sensors. The control commands, in the form of desired acceleration\footnote{A negative acceleration is interpreted by the interface as a brake command} are sent to the cars through dedicated interfaces.
%\PF{The cooperative driving system for each vehicle command the desired acceleration~\eqref{control}, where the measurements are obtained through the onboard sensors {\large \textbf{To be checked}}\footnote{A negative acceleration is interpreted by the interface as a brake command} to the interface of the pedal robot. }

\subsection{Volvo Car XC90 and Volvo Truck FH16}

\begin{figure}[!t]
\centering
\subfloat[]{
\includegraphics[width=0.25\textwidth]{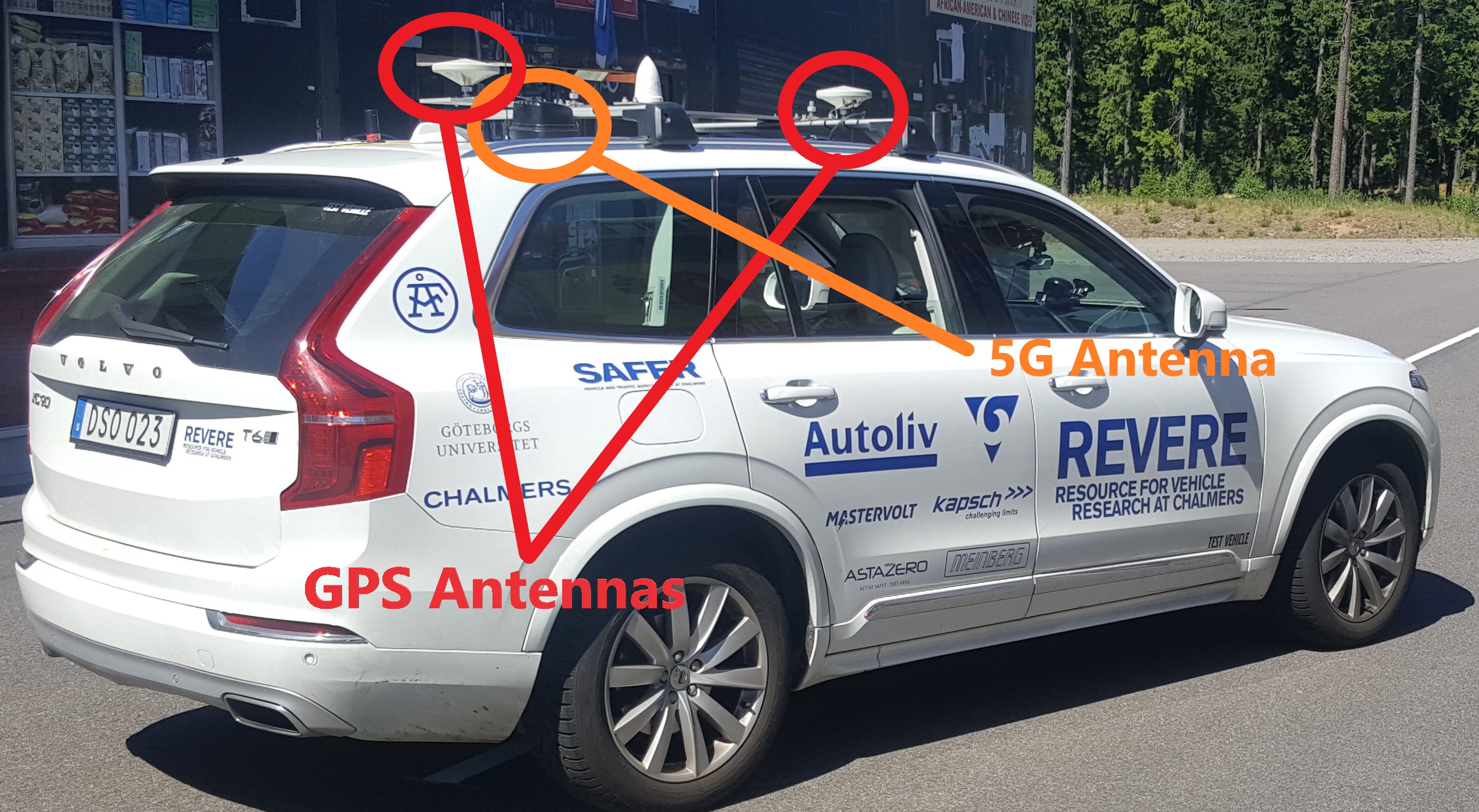}\label{fig:xc90outside}
}
\subfloat[]{ 
\includegraphics[width=0.2\textwidth]{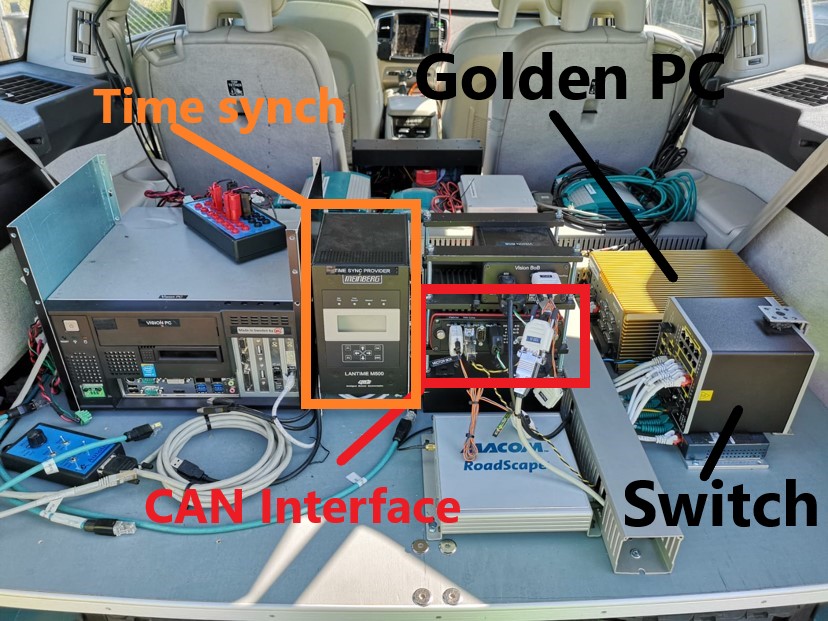}\label{fig:xc90Inside}
}
\quad
\subfloat[]{ 
\includegraphics[width=0.16\textwidth]{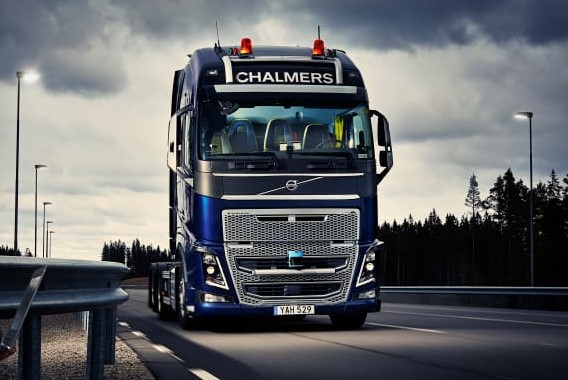}\label{fig:fh16}
} 
\subfloat[]{ 
\includegraphics[width=0.28\textwidth]{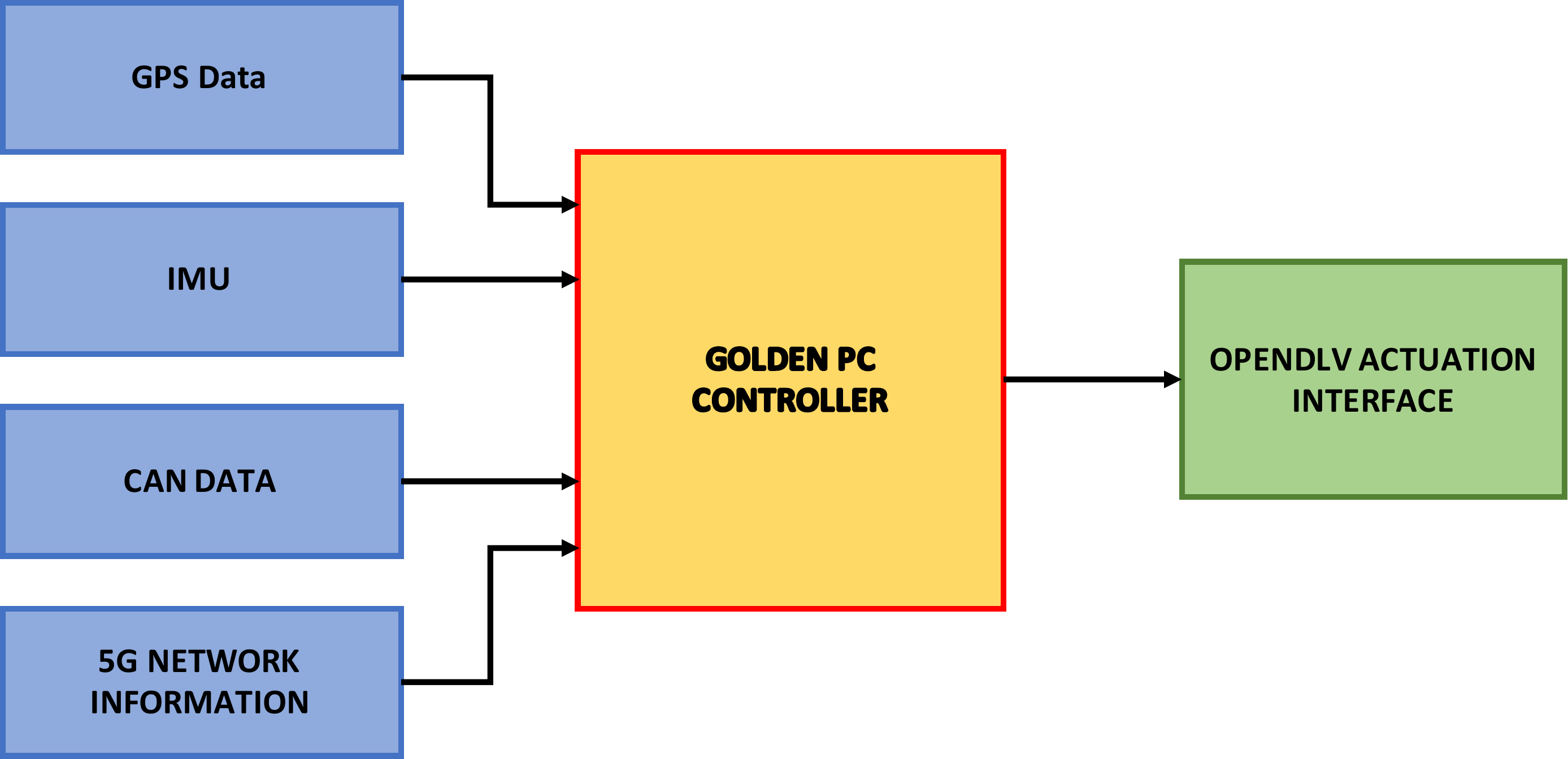}\label{architettura}
}
\caption{Experimental Setup: a) Outside  Equipment of the Volvo XC90; b) Inside Equipment of the Volvo XC90; c) Picture of the Volvo Truck FH16; d) Schematic overview of the software architecture executed on OpenDLV.}
\label{fig:testcarxc90_instrument}
\vspace{-10pt}
\end{figure}
% With respect to the cooperative driving issues, 
The XC90 is equipped as follows:
\begin{enumerate*}[label=\alph*)]
    \item Applanix GNSS/INSS unit providing the car position data in GPS coordinates. This sensor is combined with a Radio modem to gain RTK corrections, thus achieving a precision up to centimeters in data position \cite{wanninger2004introduction}.
    \item Inertial Movement Unit (IMU) providing the current vehicle acceleration. Velocity measurements are obtained from the on-board commercial ECU via the CAN Interface
    \item pre-5G Telit Modem LTE+, a 5$^{th}$ generation modem  establishing the radio communication with the Ericsson test network.
    \item Roof antennas for sharing information over the Ericsson test network.
    \item  A PC running the OpenDLV (see Section~\ref{sect:OpenDLV}) software, under a GNU/Linux based operating system (ArchLinux) processing the sensors measurements and implementing the control law in equation \eqref{control}. %It handles all the autonomous driving sensing and actuation software and computes in real-time the cooperative control input on the basis of information coming from both on-board sensing and 5G network. 
    
\end{enumerate*}
%Note that both sensing and control subsystems on the XC90 are entirely managed and supervised by OpenDLV. In particular, 
All the on-board sensors and actuators are connected, through a Local Area Netowrk (LAN), to the PC and exchange data through a UDP Multicast session. %Leveraging both data received from neighbors and on-board information, the  finite-time cooperative protocol determines the desired acceleration that has to be imposed to the vehicle. Accordingly, 
The Actuation Interface on OpenDLV provides the appropriate commands to the powertrain controller, that finally actuates the throttle and/or brake system of the XC90 (see details of the hardware configuration in Figs. \ref{fig:xc90outside} and \ref{fig:xc90Inside}; the 
software architecture, executed on OpenDLV, is instead depicted in Fig. \ref{architettura}).
With respect to the FH16 in Fig. \ref{fig:fh16}, the on-board equipment  and the software architecture executed on OpenDLV are similar to the one of the XC90 (see Figs. \ref{fig:xc90Inside} and \ref{architettura}). Indeed, the only difference is in the GNSS/INSS unit providing position data GPS coordinates, that for the truck is the Oxford OXTS GNSS (again combined with a Radio modem for RTK corrections).\\

\subsection{Volvo Car S90}

\begin{figure}[!t]
\centering
\vspace{-10pt}
\subfloat[]{
\includegraphics[width=0.1\textwidth,height=100pt]{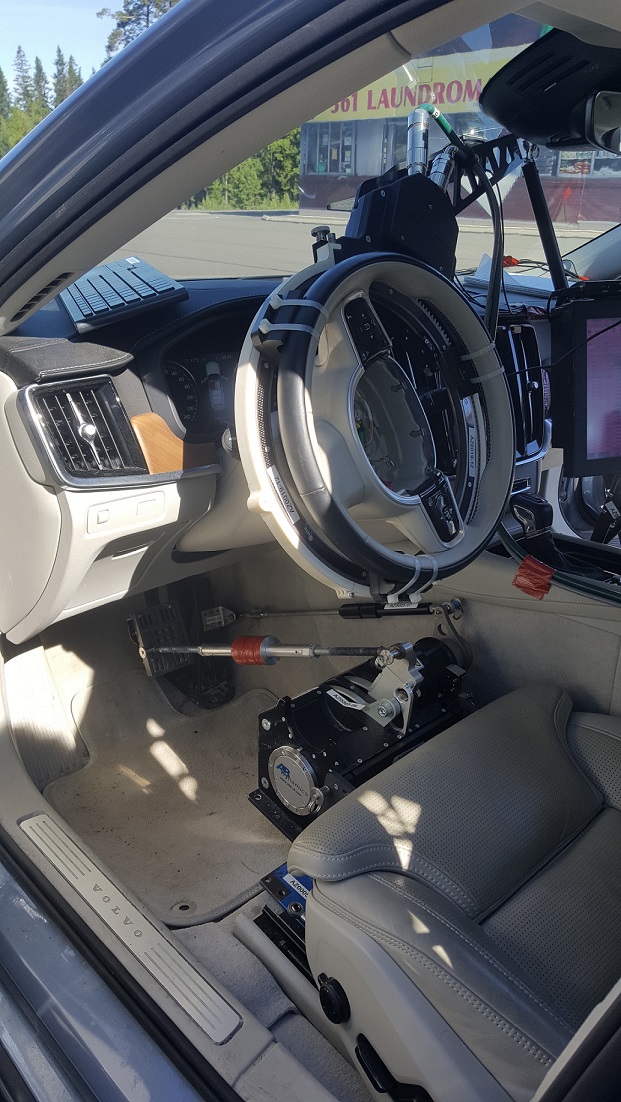}\label{fig:s90inside}
}
\subfloat[]{ 
\includegraphics[width=0.3\textwidth,height=100pt]{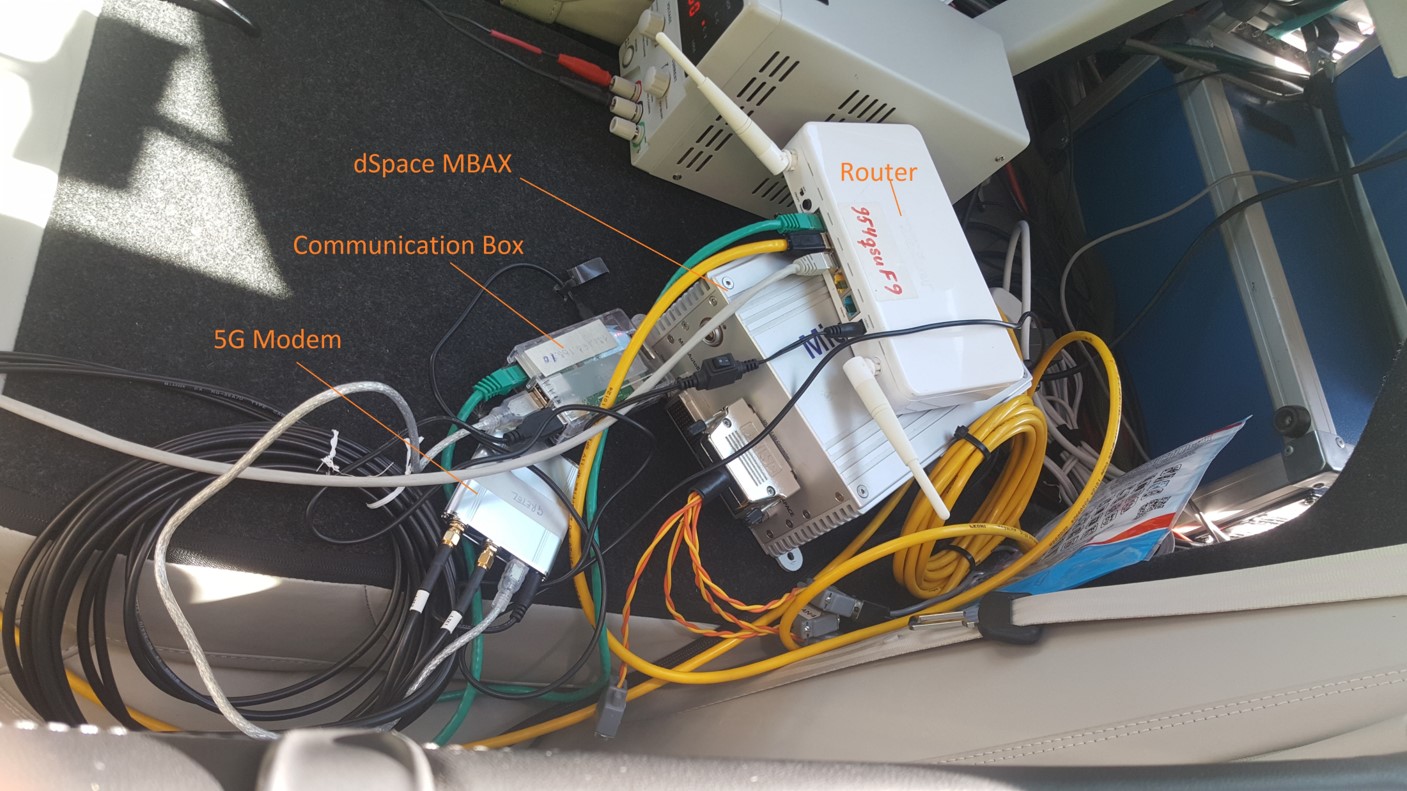}\label{fig:mbaximage}
}
\\
\subfloat[]{ 
\includegraphics[width=0.32\textwidth,height=55pt]{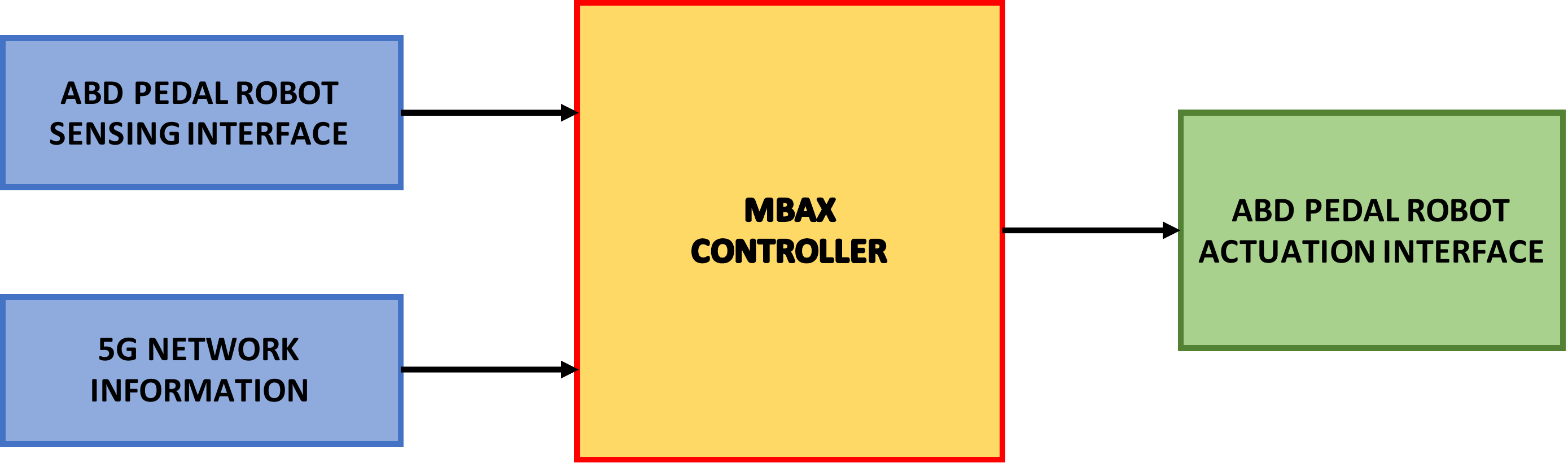}\label{architettura_DSPCACE}
}
\caption{Equipment of the Volvo S90: a) Detail of the ADB Pedal Robot; b) Details of the on-board equipment; c) Schematic overview of the software architecture executed on the dSpace MicroAutobox.}
\label{fig:testcarS90_instrument}
\vspace{-10pt}
\end{figure}

The S90 leverages the ADB Pedal Robot (shown in Fig.~\ref{fig:s90inside}) for actuating the finite-time cooperative protocol. The controller action is on-board computed via the dSpace Micro Autobox (MBAX), a real-time platform interconnected with the vehicle and the on-board equipment for cooperative driving via the Controlled Area Network (CAN) and the Local Area Network (LAN), respectively.
Namely, the ADB Pedal Robot drives the acceleration and braking systems of the vehicle (through mechanical actuators on the pedals) tracking the driving profile provided by the cooperative strategy. Moreover, the robot has a direct connection to the vehicle GNSS unit, IMU and CAN and communicates with the dSpace MBAX. The cellular communication is again guaranteed by the on-board pre-5G Telit Modem and Roof Antenna. Further specific on-board devices also include:
\begin{enumerate*}[label=\alph*]
    \item Communication Box, implemented on a  Raspberry PI that is opportunely programmed \LMC{and deployed} to receive and convert data from the Ericsson test network so that they are readable from the dSpace MBAX. 
    \item On-board Switch for providing the in-vehicle LAN.
\end{enumerate*}
Details of both hardware configuration and software architecture, executed on MBAX, are shown in Figs. \ref{fig:mbaximage} and \ref{architettura_DSPCACE}, respectively.
Specifically, the MBAX Controller has been prepared to run a Matlab/Simulink schema whose main aim is to gather information about the state of vehicles, merge it with traffic data coming from the cloud and compute a control output for the \textit{ADB}.
dSpace MBAX operation can be hence summarized as follows: $i$) receiving current states of Volvo S90 from the \textit{ADB}; $ii$) gathering traffic information from the communication box; $iii$) computing the control input and sending actuation signals to the \textit{ADB}; $iv)$ communicating current known states, through the communication box, to the cloud.

\subsection{OpenDLV Communication Module}
\label{sect:OpenDLV}
OpenDLV is a modern open source software environment to support the development and testing of self-driving vehicles. It has been implemented using high quality and modern C++14 with a strong focus on code clarity, portability, and performance. In addition, it is entirely based on micro-services, usually run in separate docker~\cite{merkel2014docker} containers.
%The structure of this environment is just briefly summarized in this work.
For a more comprehensive treatment, we refer the interested reader to the specific literature (\cite{berger2016open} \cite{berger2017containerized} \cite{benderius2018best}). 
%Concisely, within the OpenDLV environment, each physical component (sensor or actuator) has a virtualized counterpart, on the main computer where the OpenDLV software runs. Therefore, the virtual counterpart, running as a docker image, is the only entity allowed to establish an actual communication channel with the represented physical component. All the messages, either coming from a physical component or going towards it, shall be exchanged through the associated virtual component. In particular, each virtual docker image exchanges data with the physical component it is associated with through a dedicated protocol (e.g., the Applanix OpenDLV component uses TCP/IP to get data from the Applanix GPS). Information received from the physical component is then sent to the LAN as a UDP-based multicast message. In addition, all the messages exchanged within OpenDLV sessions share a common structure defined in the \emph{OpenDLV Standard Message Set}. Therefore, OpenDLV enables interoperability between hardware and software interfaces by decoupling high-level application logic from low-level device drivers. 
%\begin{figure}
%    \centering
%    \includegraphics[width=\textwidth]{images/Chap4/xc90vp.PNG}
%    \caption{Volvo XC90 - OpenDLV Architecture}
%    \label{fig:xc90vp}
%\end{figure}
%As anticipated before,
In  order to extend the communication abilities of our test cars to $5$G, an additional module has been designed and developed, in accordance with the OpenDLV specification. The \textit{OpenDLV Standard Message Set} has been hence expanded to account for sensing information from other vehicles. %as showed in \ref{lst:comm_message_set}. 
Among the added message properties, we find the number of connected nodes (i.e., the number of current active nodes at the intersection), the sequence number of the packet and a so-called \textit{whoami} field that each vehicle uses to identify itself %inside
\LMC{within} a fleet. %Moreover, it can be noticed that the allocation of the list of the vehicles is static rather than dynamic. This has been done because OpenDLV Standard Messages Set did not support dynamics list of \textit{ODVD} during the development of this communication tool.
The sequence number is used to discard packets received out of sequence. The OpenDLV receiver keeps track of received packets: if the current received packet has a lower sequence number than the last packet received, it will be discarded and won't be replayed in the UDP-based OpenDLV session.
The OpenDLV communication module has been thought to run in a container within an OpenDLV session. In particular, this module is able to exchange data with other OpenDLV components such as the \textit{proxy interface} to the car CAN bus and the \textit{Applanix GPS}. 

\subsection{Hermes Traffic Monitor}
A web interface has been built that allows a user to monitor the status of the traffic, by providing details about all connected vehicles. Thanks to the generalized structure of the communication software and the standardization of most of the \textit{traffic management} system, the development of this component took a minimal amount of effort and time.

\section{Experimental Validation}
\label{sec:experiments}

\subsection{Illustrative Driving Scenario}\label{descrizione}
%{\bf STEFANIA}
The tests have been executed at the City Area of the AstaZero\footnote{http://www.astazero.com} (near Gothenburg, Sweden).
The City area\footnote{http://www.astazero.com/the-test-site/test-environments/city-area/} consists of small town centers with streets, of varying widths and lanes, equipped with bus stops, pavements, street lighting, and building backdrops. The road system, allowing different kinds of driving tests, includes roundabouts, T-junctions and return-loops. Connections to the rural road occur in two places. The area has a relatively flat surface with dummy blocks that resemble buildings and host some technical aids such as radars (see Fig.  \ref{foto_cityarea}). One of the blocks also contains space for a control room and a warehouse for dummies.

\begin{figure}[!ht]
    \centering
    \includegraphics[width=0.75\columnwidth]{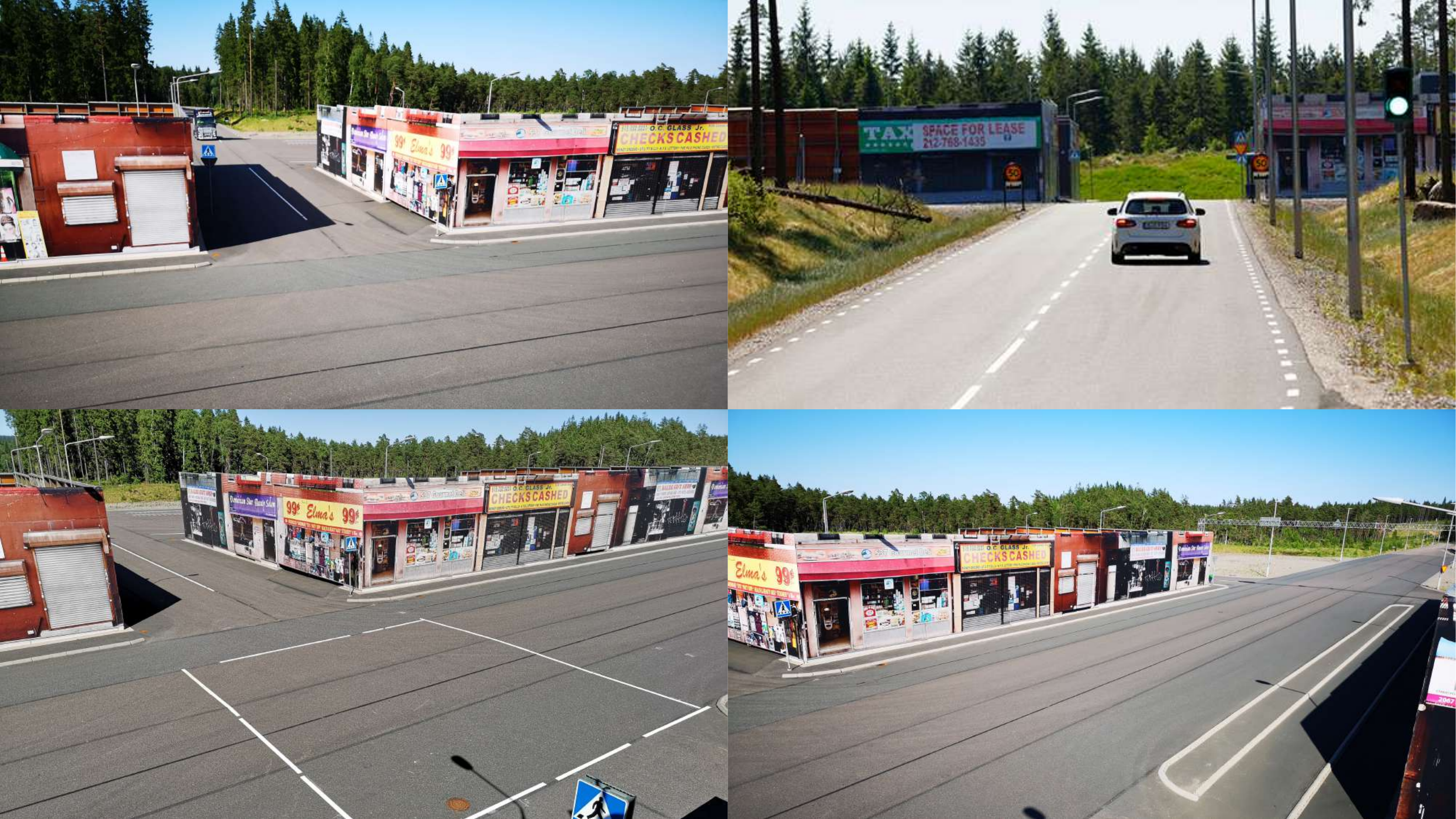}
    \caption{The City Area at AstaZero}
    \label{foto_cityarea}
       %\vspace{-10pt}
\end{figure}

\begin{figure}[!ht]
    \centering
    \vspace{-10pt}
    \includegraphics[width=0.85\columnwidth]{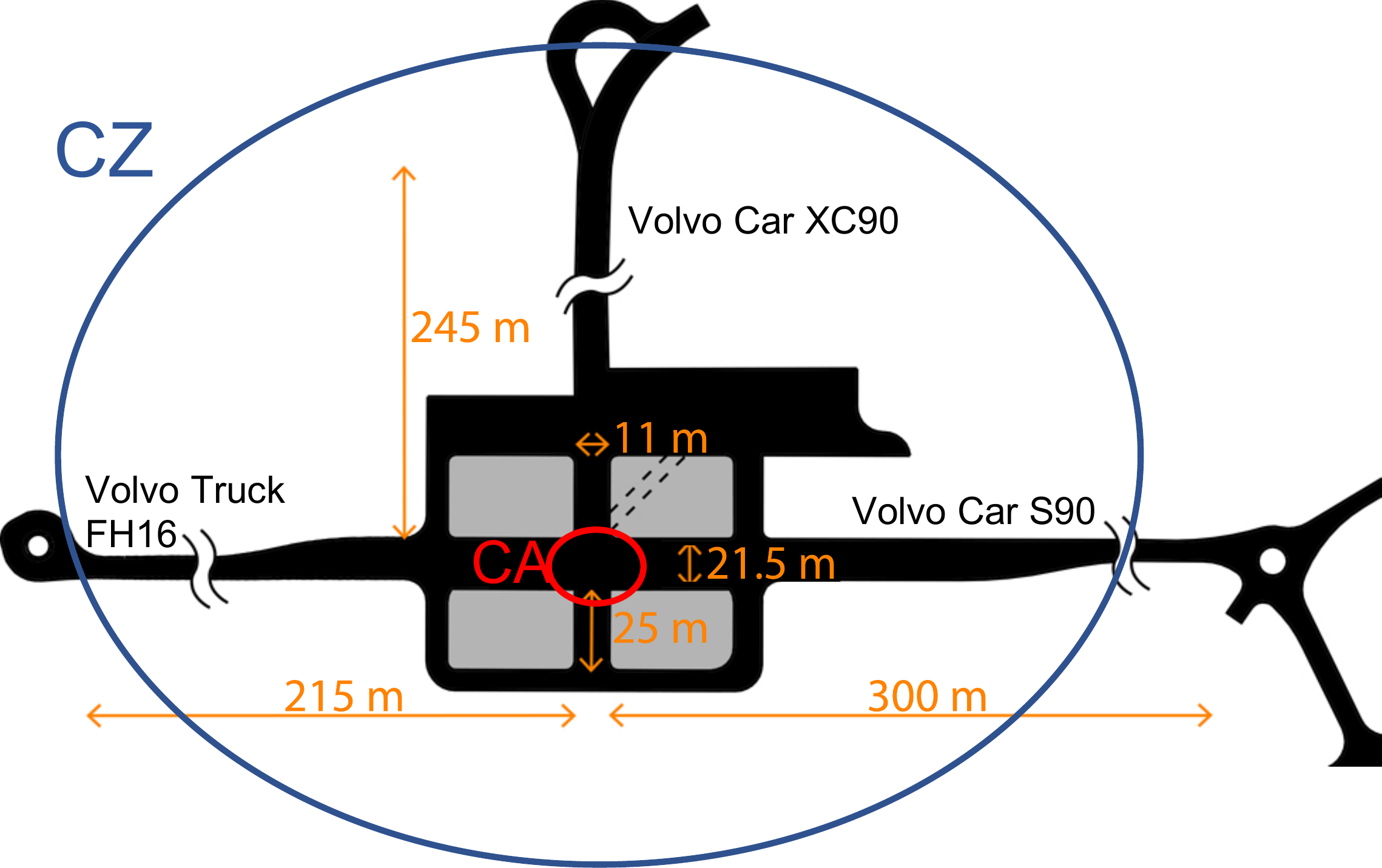}
    \caption{Map of the City Area at AstaZero}
    \label{fig:cityareamap}
       %\vspace{-10pt}
\end{figure}

The map of the City Area exploited for the tests is reported in Fig. \ref{fig:cityareamap}. Here, the Cooperative Zone (CZ) of interest is marked with a blue circle, while the Conflicting Area (CA) (at the intersection center) is marked with a red circle. 
Different experimental runs were performed in different driving conditions. 
In what follows we will first describe how we performed preliminary trials in a so-called \emph{multilane} scenario allowing us to safely simulate a real-world intersection. We will then move to the actual street junction scenario, for which we will report some of the experimental results related to the case when the vehicles, initially located as in Fig. \ref{fig:cityareamap}, access the CZ with initial velocities and relative positions that would lead to collision without any control action. 
This exemplar scenario also considers mixed traffic. Namely, the Volvo XC90 and S90 are fully automated, while the Volvo Truck F16 is human-driven, but connected, i.e., it shares information about its actual position and speed. Note that mixed traffic situations are the ones that at first will arise in the very next future when the fully autonomous and the human-driven cars will interact on the road via a V2V communication network.

\subsection{Experimental Campaign}
Experimental validation is carried out via both multilane and street Junction experiments.
A multilane experiment is a safe real emulation of a street intersection where roads leading to the intersection are projected parallel to each other. Involved vehicles will drive in parallel while approaching a designated area they must access into, according to a mutual exclusion policy. %(see Fig.~\ref{fig:multilane}). 
 This scenario is of the utmost importance since it allows to test control algorithms in a highly realistic situation, taking account of the actual delay introduced by centralized communication hardware and software, without the risk that vehicles will collide.

%The experiment takes place all over the main road of the City Area. %(215+300 meters) \ref{fig:cityareamap} and its overview is showed in \ref{fig:cityareaback}. Please, note that the sequence of vehicles crossing the intersection has not been necessarily the same showed in the figure. 

%\begin{figure}
%    \centering
%    \includegraphics[width=.95\columnwidth]{imgs/multilane.jpg}
%    \caption{Multilane Scenario}
%    \label{fig:multilane}
%       \vspace{-10pt}
%\end{figure}

%\subsection{Street Junction Experiment}

Once results have been validated with the multilane experiments, they can easily be replicated on a real intersection without risking collisions. It must be highlighted that buildings placed at the corners of the intersection constitute an obstacle against both the human driver's eye and a virtual direct \emph{Vehicle-to-Vehicle} communication. 
\begin{figure*}[!t]
\centering
\vspace{-10pt}
\subfloat[]{\includegraphics[width=0.24\textwidth]{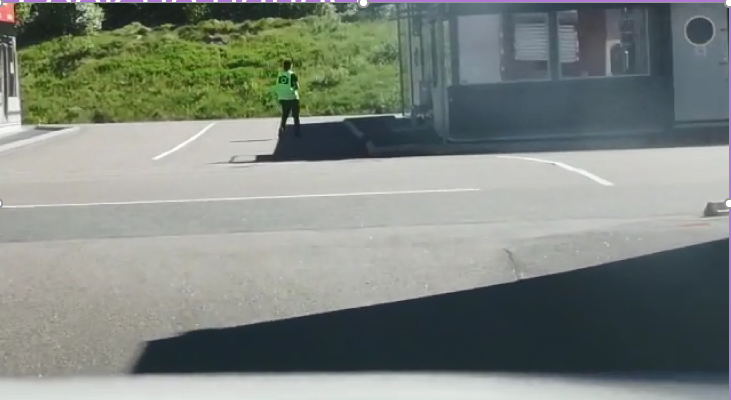}\label{fig:saved_1}}
\subfloat[]{ \includegraphics[width=0.24\textwidth]{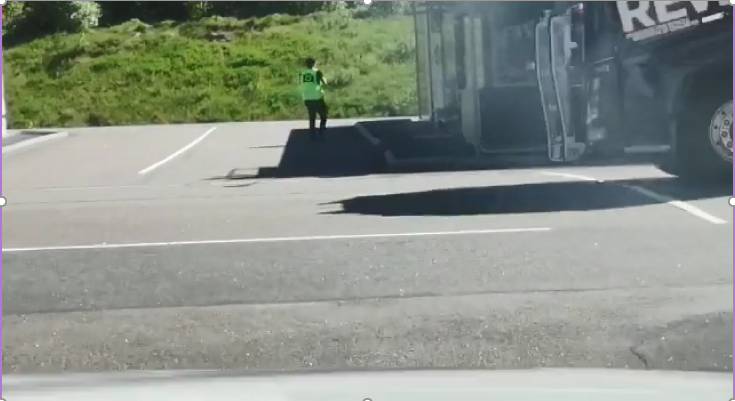}\label{fig:saved_2}}
\subfloat[]{\includegraphics[width=0.24\textwidth]{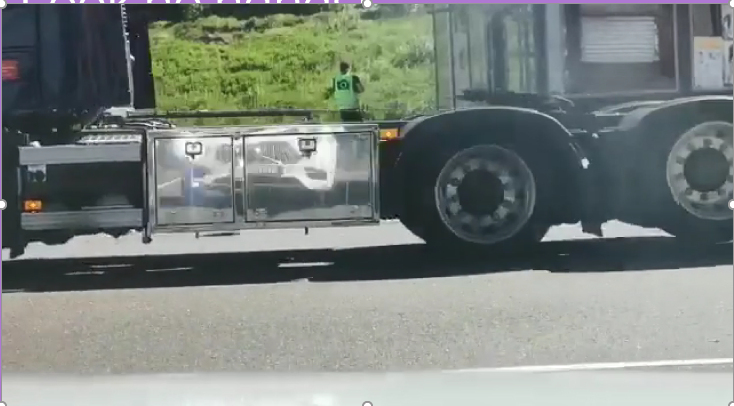}\label{fig:saved_3}}
\subfloat[]{ \includegraphics[width=0.24\textwidth]{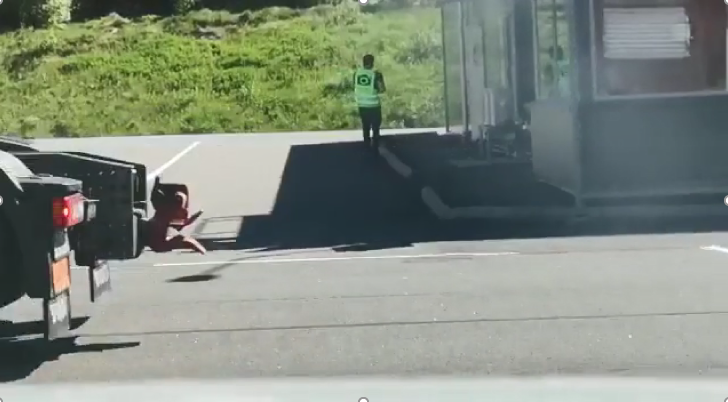}\label{fig:saved_4}}\\
\caption{First Person View from Volvo Car XC90: the truck stays hidden until the very last second due to the shape of the urban area.}
\vspace{-10pt}
\label{fig:saved}
\end{figure*}
Under this assumption, in fact, each vehicle would not be able to see hidden mobile nodes %(i.e., Volvo XC90)
approaching the crossroad until the very last moment, without being connected to the 5G PoC. Our tests indeed show that it is still possible to avoid collisions and reach a consensus thanks to a high speed (5G) centralized connection, despite Volvo XC90 staying hidden from Volvo Truck FH16 and Volvo Car S90 (and viceversa) for almost the entire duration of the experiment (as demonstrated by the sequence of snapshots in Fig.~\ref{fig:saved}).

%\begin{figure}
 %   \centering
  %  \includegraphics[width=.95\columnwidth]{imgs/intersection.jpg}
  %  \caption{Intersection Scenario. \bf{Stefania: Toglierei la figura perchè è già presente 18}}
 %   \label{fig:intersection}
%\end{figure}

\subsection{Outcomes}

In both classes of experiments, multilane and intersection, the overall system has successfully demonstrated its capacity of managing the negotiation of street junctions over pre-5G PoC, using LTE radio with 5G EPC. In the following sections, the results obtained in the real intersection scenario will be illustrated and discussed. We discuss the results for just one class of experiments since the two classes are identical from a scientific point of view. The choice of executing the multilane set of trials before moving to the real intersection scenario just depends on reasons related to the safety of people inside the cars.

\section{Network Performance}
\label{sec:netperformances}

%{\bf @SIMON, QUESTA PARTE E? FONDAMNTALE PER TON E MANCA.
%Discussione dei tempi ottenuti durante gli esperimenti. La compatibilita' dei tempi ottenuti con i vincoli di controllo imposti dagli algoritmi di controllo potrebbe essere usata come trampolino la subsection successiva.}

\LMC{The proposed model and framework (Sections ~\ref{sec:problemstatmenet} and ~\ref{sec:ControlProtocol}) have been successfully implemented, deployed and tested on three vehicles at the AstaZero (AZ) proving ground. %Outcomes have been outlined at the end of Section~\ref{sec:experiments}.
Communication among entities has been enabled by the pre-5G PoC test network provided by Ericsson. The pre-5G PoC, using LTE radio with 5G EPC, offers a full radio coverage of the City Area (CA) of AstaZero, as outlined in Fig.~\ref{fig:cityareamap}. In order to reduce control and management times, the distributed cloud network is installed within the boundaries of the proving ground itself. Experiments discussed in this work date back to March 2018 and the Ericsson test network has evolved meanwhile towards 5G NR.} %as showed in Table~\ref{tab:5Gtimeline}, %from  reports the network development time line, internally released by Expericom.}

\subsection{Preliminary Latency Analysis}
\label{sec:preliminaryanalysis}

%\begin{figure}[!h]
%  \centering
%  \includegraphics[width=.6\columnwidth]{imgs/Stevenson.png}
%  \caption{Preliminary Latency Analysis: Stevenson diagram.}
%  \label{fig:stevenson}
%  \vspace{-10pt}
%\end{figure}

\LMC{In early March 2018, performance of the network has been measured in terms of delay and latency. Early measurements have been carried out with the aim of understanding the impact of the network on the overall communication performance. The analysis enabled us to design a communication software that could meet the safety timing constraints required by the involved control algorithms. The network flow between two laptops connected to the test network showed an average TCP Round Trip Time (RTT) of $24ms$ with Standard Deviation (STD) of $0.028$. The RTT has been measured using the TCP Acknowledgement segment. In addition, the measurements of the time between two consecutive frames showed an inter-frame latency of $24ms$, along with $0.027$ standard deviation. We need to point out that measurements of our interest have been taken at software level, and that the air-interface latency is a small part of the RTT we measured. Finally, a sequence number analysis has been conducted via the Stevenson graph (omitted for sake of brevity) suggesting
%reported
%in Fig.~\ref{fig:stevenson}
%shows how the sequence number of the packets grows in time.
%Indeed, the line growing linearly without flattening,
that there are no relevant packet delays in the communication.
From this analysis it is reasonable to assume that the performance offered by the network, including delay, meets control constraints and potentially enables a safe and reliable communication between a set of mobile nodes and a remote endpoint on the ground.}
%Initial measurements of network performance and latency were carried out at the beginning of March. The main purpose of this phase has been to get a big picture of the involved times in order to design reliable and safe communication algorithms for time management with low latency.
%\begin{figure}[ht]
%  \centering
%  \includegraphics[width=.8\columnwidth]{imgs/Combined.png}
%  \caption{TCP Flow \PF{{\large Fonts are too small}}}
%  \label{fig:tcpflow}
%  \vspace{-10pt}
%\end{figure}
%The diagrams summarized by Fig.~\ref{fig:tcpflow} show an average Round Trip Time (RTT) of $24ms$ with Standard Deviation (STD) of $0.028$. The RTT has been measured using the TCP Acknowledgement segment. In addition, the measurements of the time between two consecutive frames showed an inter-frame latency of $24ms$, along with $0.027$ standard deviation. Also, the Stevenson diagram reported in Fig.~\ref{fig:stevenson} shows how the sequence number of the packets grows in time. The fact that the line grows linearly and does not flatten, suggests that there are no relevant packet losses in the communication. In case of relevant packet losses, indeed, the line would have grown following an exponential shape. From this analysis, carried out at TCP level,  it is reasonable to assume that the latency offered by the network, along with the low loss rate, is good enough to enable a safe and reliable communication between a set of mobile nodes and a remote endpoint on the ground.

\subsection{Delay components}
%To study the behaviour of the network in terms of reliability and latency of the connection, three different types of delay have been taken into account. 
\LMC{Three kinds of delays have been taken into account when measuring the performance offered, through the pre-5G PoC, by the communication software designed for our vehicles}. Firstly, since the conceived \textit{Traffic Management} protocol has been designed to rely on TCP, we wanted to measure the impact of the chosen transport protocol on the overall communication. Moving up along the ISO/OSI stack, two additional delay contributions have been considered, namely the \textit{Application Layer ACK} (WS ACK) and the so-called \textit{State RTT}, as outlined in Figure \ref{fig:delay_composition}. The former is the acknowledgement time of a websocket packet at the application level.  The latter measures the interval between the time a vehicle sends its state and the time it receives the same state reflected from the traffic manager, in the form of a \textit{traffic update} message. The \textit{State RTT} is of utmost importance, since it represents the delay that is actually impacting the control algorithm. Also in this case, we want to remark that measurements of our interest have been taken at software level, and that the air-interface latency is a small part of the RTT we measured. 

\begin{figure}[!t]
    \centering
    \includegraphics[width=1\columnwidth]{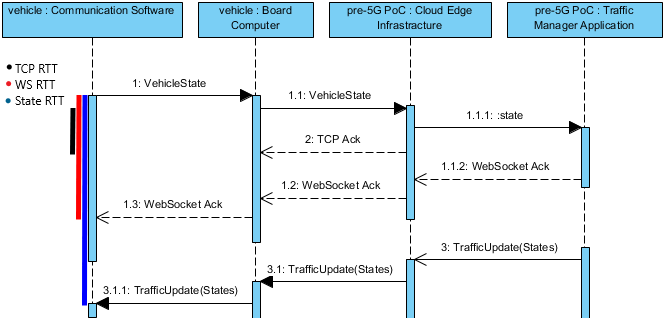}
    \caption{Contributions to measured delays}
    \label{fig:delay_composition}
\end{figure}

\subsection{TCP Analysis} 

\begin{figure*}[!t]
\centering
\vspace{-10pt}
\subfloat[]{\includegraphics[width=0.253\textwidth]{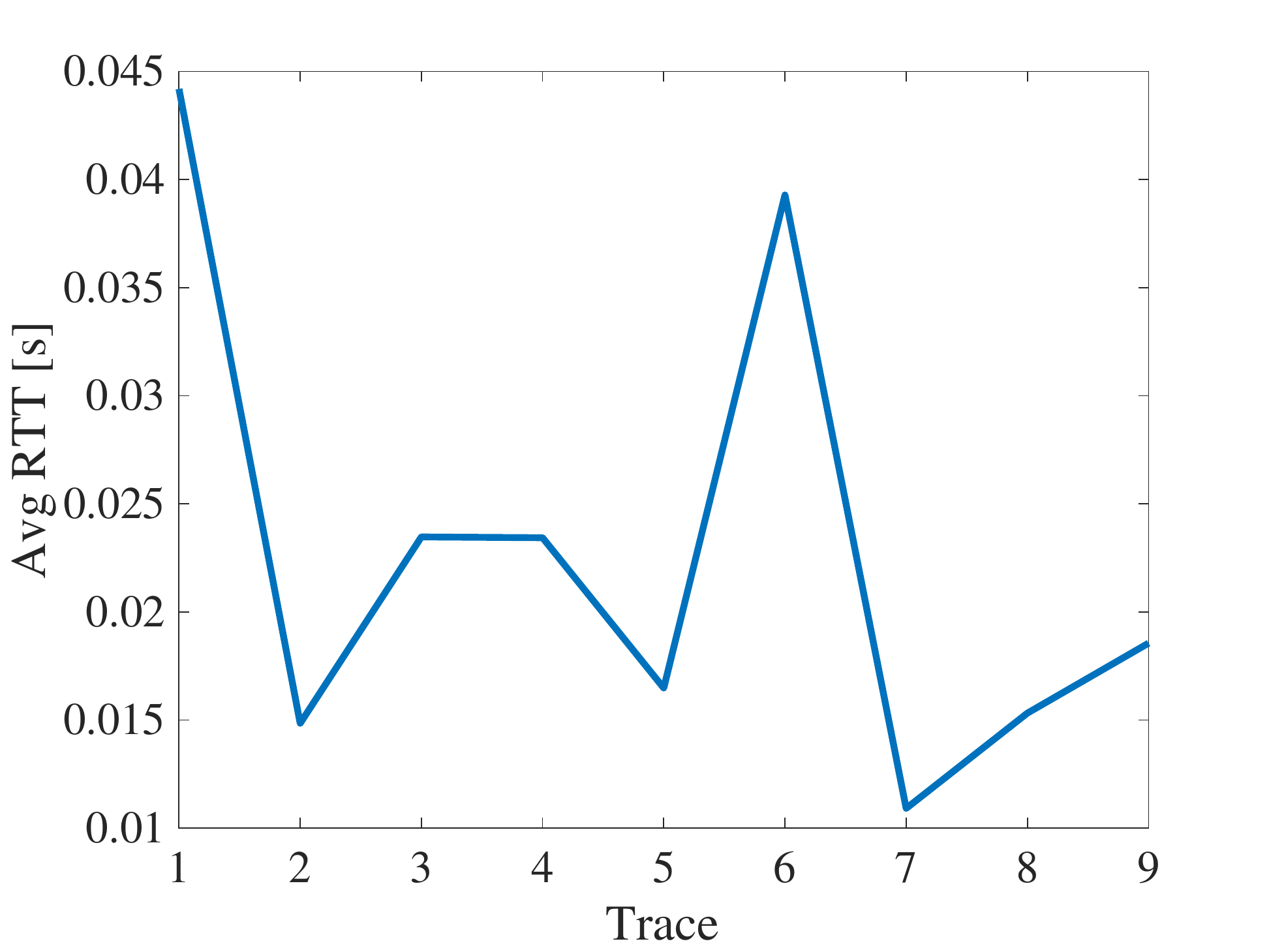}\label{fig:av1}}
\subfloat[]{ \includegraphics[width=0.253\textwidth]{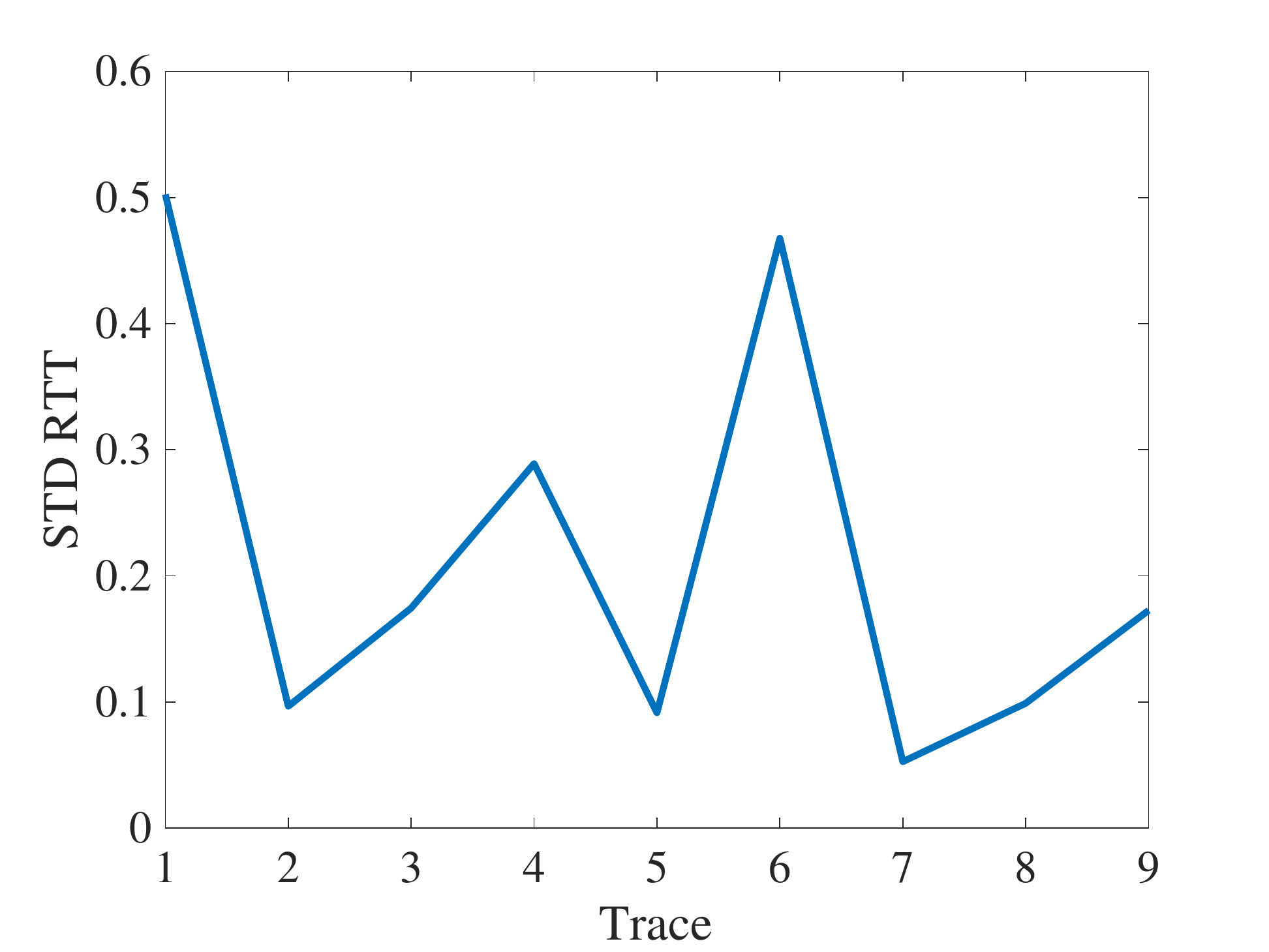}\label{fig:std1}}
\subfloat[]{\includegraphics[width=0.253\textwidth]{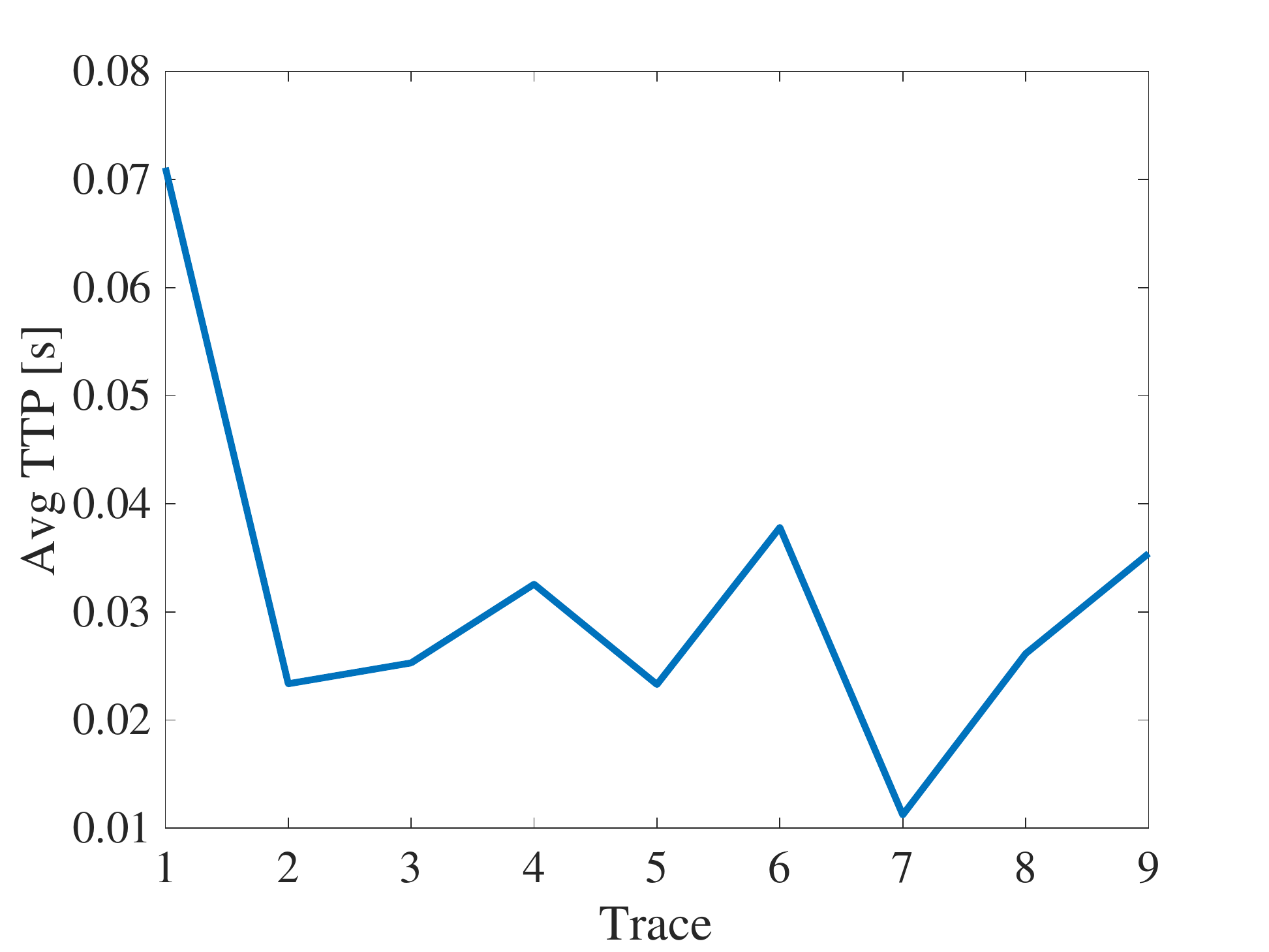}\label{fig:av2}}
\subfloat[]{ \includegraphics[width=0.253\textwidth]{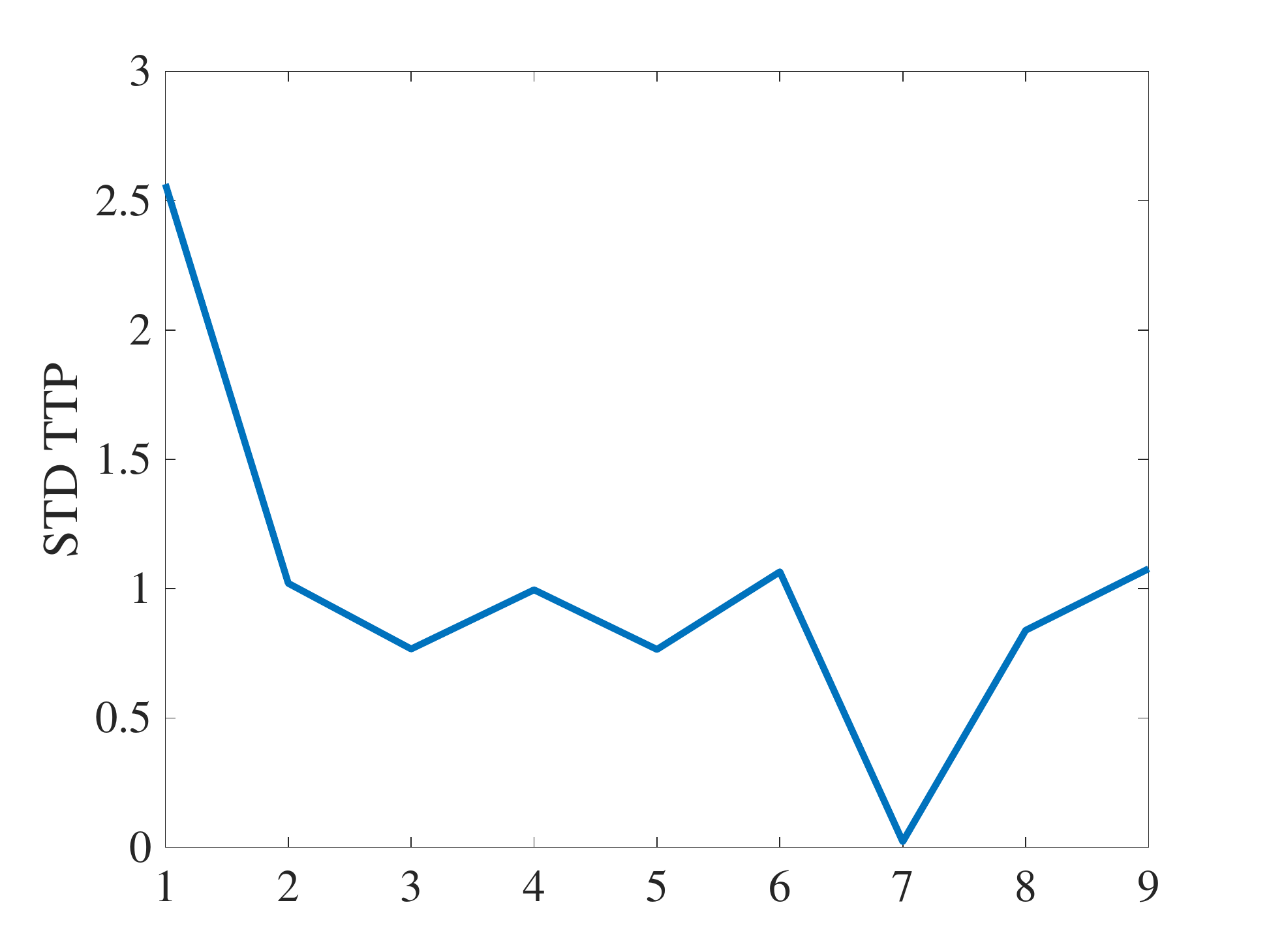}\label{fig:std2}}
\caption{Nine TCP traces recorded during the experiments:
a) average TCP Round Trip Time (RTT);
b) RTT Standard Deviation normalized to $N_{sample} -1$;
c) average Time From Previous Frame (TTP);
d) TTP Standard Deviation normalized to $N_{sample} -1$.}
\vspace{-5pt}
\label{fig:net_summary}
\end{figure*}

\LMC{During the experiments, nine TCP traces have been collected. As envisaged, their analysis shows very low delays along the whole experimental session. In particular, two features have been considered at this stage: Round Trip Time (RTT) and Time From Previous Frame (TTP), whose aggregated behavior during the entire duration of the experiments is reported in Fig.~\ref{fig:net_summary}. A more detailed diagram has been reported for Traces $1$ and $8$ in Fig.~\ref{fig:tcpdelay}. With respect to the \textit{RTT}, the average value oscillates between $15$ and $45$ ms. Furthermore, the average time between two subsequent TCP frames received oscillates between $10$ and $40$ ms with a peak of $70$ ms on the first trace. It is important to point out that the nine traces differ between one another in number of packets. In particular, trace $1$ stores a number of packets whose order of magnitude is $10^5$, while the others store a number of packets with orders of magnitude between $10^3$ and $10^5$. A noticeable improvement has been observed since the preliminary tests, due to hardware and software improvements that have been applied to the test network. Outliers in the diagrams are negligible as their number is infinitesimal compared to the overall number of transmitted packets. Reliability of the system is hence not impacted. To this extent, the Stevenson Diagram for Trace 1, omitted for sake of brevity, does not show any major issues related to communication.}
%reported in Fig.~\ref{fig:real_stevenson},

%\begin{figure}[!ht]
%    \centering
%    \begin{subfigure}[b]{0.45\columnwidth}
%        \includegraphics[width=\columnwidth]{imgs/real_acktime.png}
%        \caption{RTT time to ACK}
%        \label{fig:tcpdelay1}
%    \end{subfigure}
%    ~
%    \begin{subfigure}[b]{0.45\columnwidth}
%        \includegraphics[width=\columnwidth]{imgs/real_sincepreviousframe.png}
%        \caption{Time since previous TCP frame}
%        \label{fig:tcpdelay2}
%    \end{subfigure}
%    ~ 
%    \caption{TCP Delays on intersection negotiation}
%    \label{fig:tcpdelay}
%\end{figure}
%\begin{figure}[!t]
%\centering
%\subfloat[]{\includegraphics[width=0.4\textwidth]{imgs/real_acktime.png}\label{fig:tcpdelay1}}\\
%\subfloat[]{ \includegraphics[width=0.4\textwidth]{imgs/real_sincepreviousframe.png}\label{fig:tcpdelay2%}}
%\caption{TCP Delays on intersection negotiation}
%\vspace{-10pt}
%\label{fig:tcpdelay}
%\end{figure}
\begin{figure}[!t]
\centering
\vspace{-10pt}
%\subfloat[]{\includegraphics[width=0.25\textwidth]{imgs/network_new/1RTT.jpg}\label{fig:rtt1}}
\subfloat[]{ \includegraphics[width=0.25\textwidth]{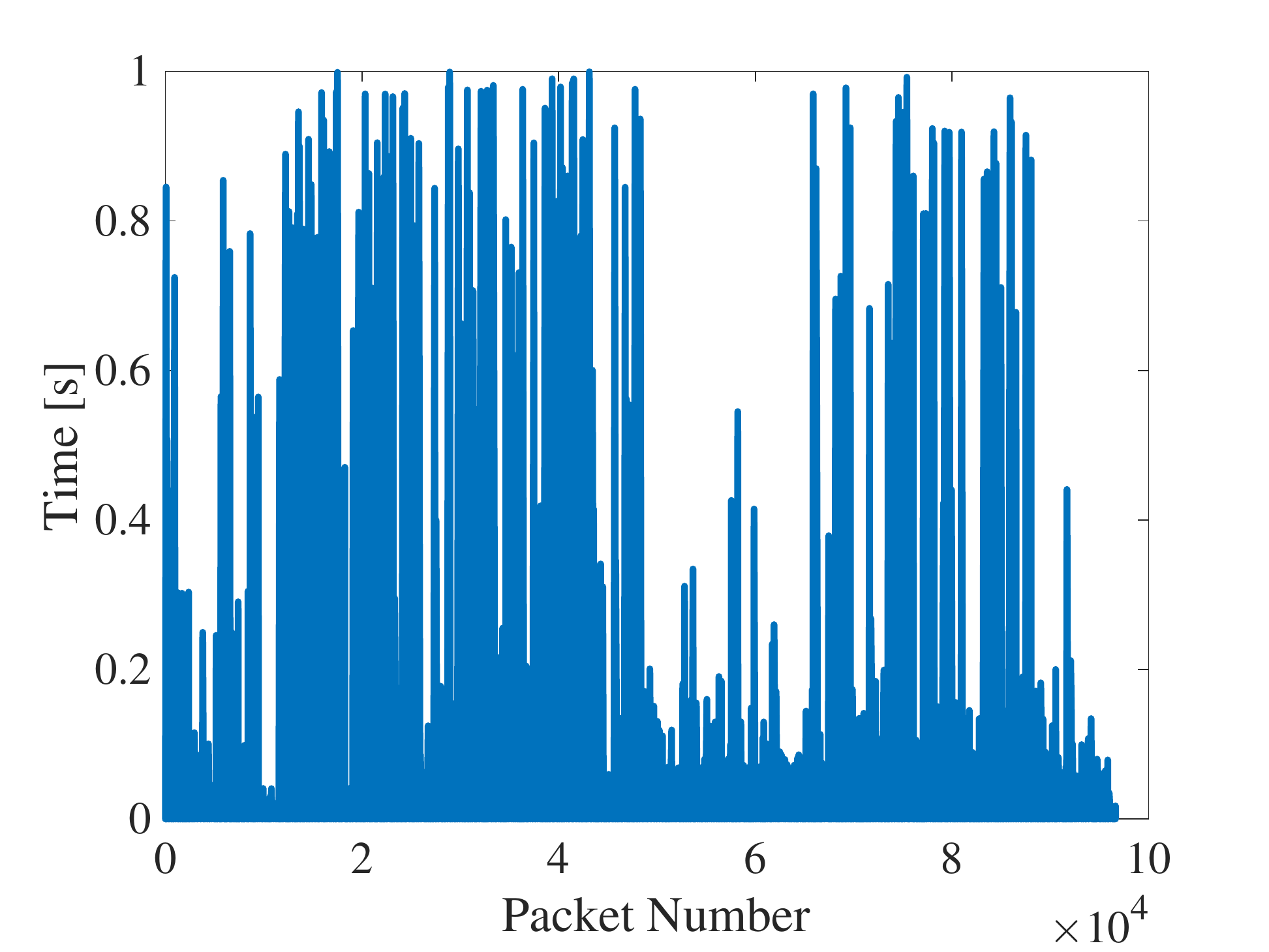}\label{fig:rtt12}}
%\subfloat[]{\includegraphics[width=0.25\textwidth]{imgs/network_new/8RTT.jpg}\label{fig:rtt2}}
\subfloat[]{ \includegraphics[width=0.25\textwidth]{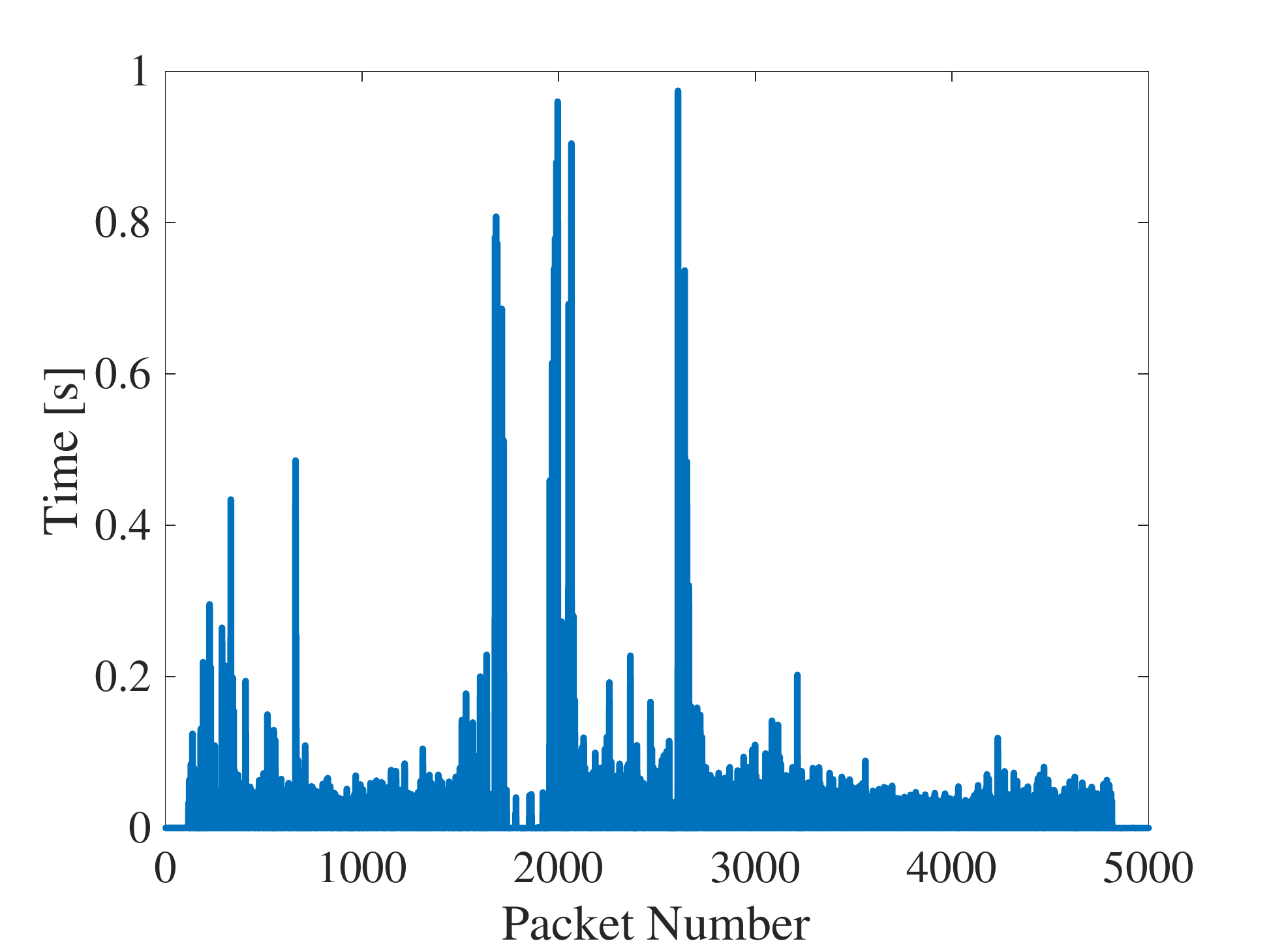}\label{fig:rtt22}}
\caption{Round Trip Time TCP Delay during two of the nine TCP traces we have recorded: a) Trace 1; b)Trace 8.}
\vspace{-10pt}
\label{fig:tcpdelay}
\end{figure}
%\begin{figure}[!ht]
%    \centering
%    \vspace{-10pt}
%    \includegraphics[width=0.8\columnwidth]{imgs/real_app_performances.png}
%    \caption{State RTT for one single experiment}
%    \label{fig:staterttonesingleexperiment}
%        \vspace{-15pt}
%\end{figure}
\begin{figure}[!t]
\centering
\vspace{-10pt}
\subfloat[]{\includegraphics[width=0.8\columnwidth]{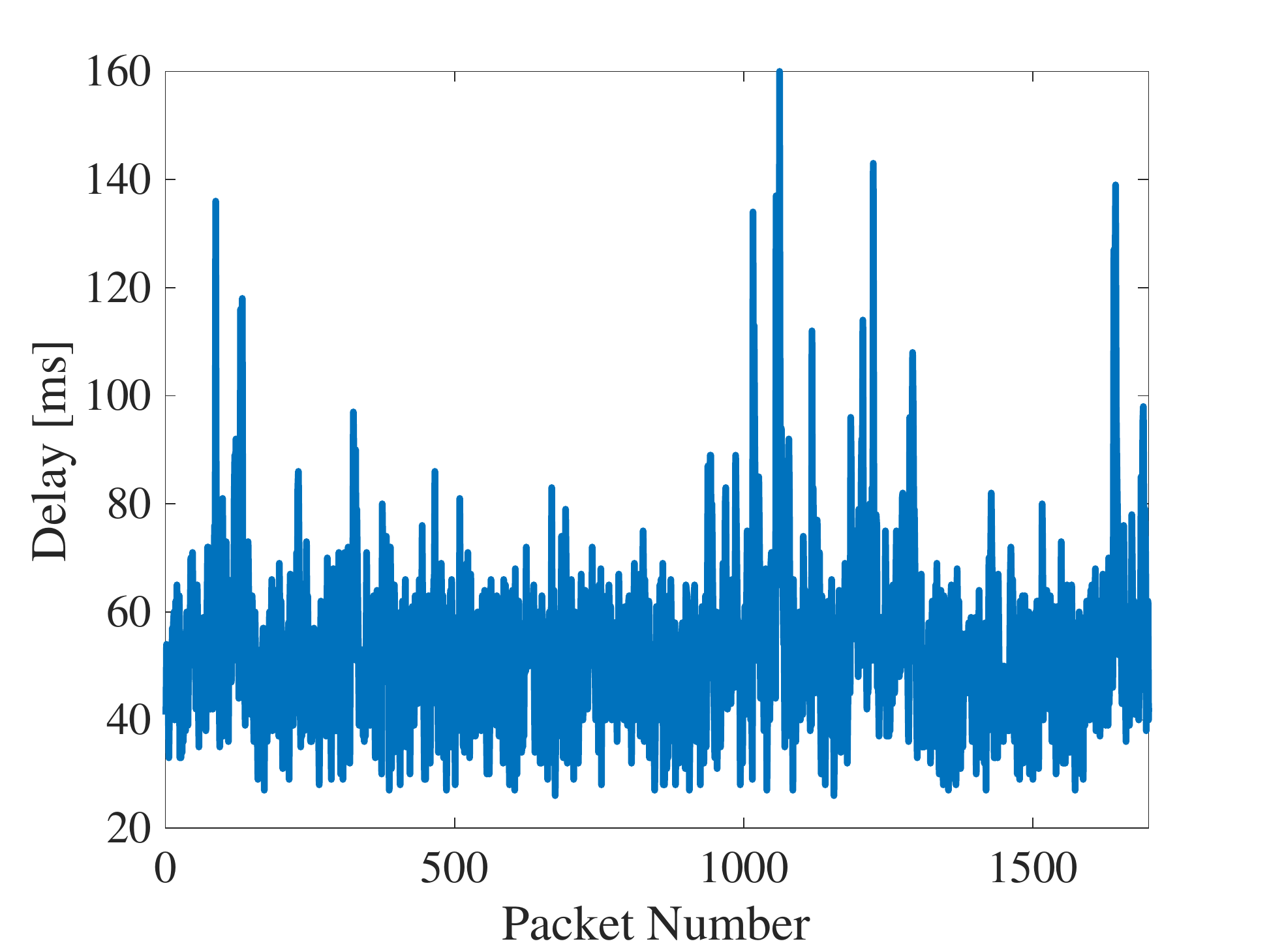}\label{fig:staterttonesingleexperiment}}\\
\vspace{-10pt}
\subfloat[]{\includegraphics[width=0.4\textwidth]{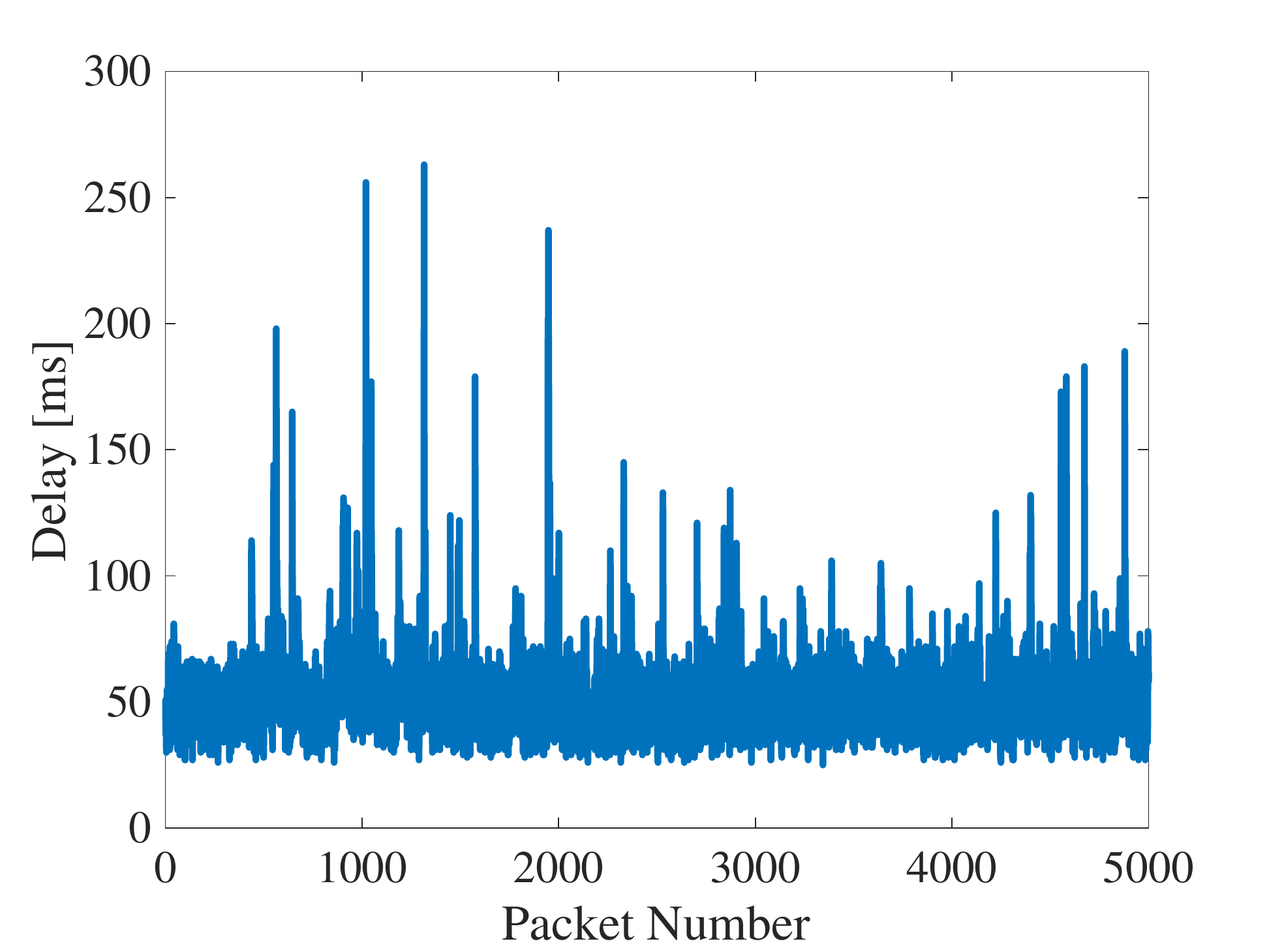}\label{fig:statertt}}
\vspace{-7pt}
%\subfloat[]{ \includegraphics[width=0.35\textwidth]{imgs/network_new/STATERTT_ZOOM.jpg}\label{fig:statertt_zoom}}
\caption{State RTT:
a) for one experiment;
a) for a series of experiments conducted in the real intersection scenario.
%b) zoom on \cref{fig:statertt}
}
\vspace{-15pt}
\label{fig:appdel2}
\end{figure}

%\begin{figure}
%    \centering
%    \includegraphics[width=.9\columnwidth]{imgs/real_stevenson.png}
%    \caption{TCP Trace 1: Stevenson Diagram during experiments}
%    \label{fig:real_stevenson}
%        \vspace{-10pt}
%\end{figure}
%\subsection{Application Level ACK RTT}
\subsection{Communication Software Performance}
As mentioned above, one parameter for classifying the performance perceived at application level is the RTT ACK. %we have also measured the ACK perceived at application level. 
Measurements report a value which is bounded within the interval $\left[0-50\right]$ ms . As a matter of fact, these data are quantitative and cannot be considered $100\%$ reliable ($0$ seconds delay is not realistic). However, they provide coarse-grained information about the order of magnitude of the measured indicator. Such lack of precision can be ascribed to two main factors:
a) \emph{Software}: The acknowledgement time has been measured, on the vehicles, via the \texttt{ack-callback} provided by the C Websocket Library, which is known not to be triggered in a strictly real-time fashion;
b) \emph{Operating System}: The acknowledgement is measured at a really high level in the operating system.
\LMC{With the above considerations in mind, we can nonetheless safely state that the application Level RTT ACK stays below the envisaged design threshold.
On the other hand, the parameter used to measure the actual delay impacting on the control algorithm is the \textit{State RTT}. It is indeed, for a vehicle, the window between a \emph{status message} sent to the traffic manager and the very same message received as \emph{traffic information} from the traffic controller. This delay is measured at application level and it is the ultimate, composed, delay that might affect the control algorithm. The control action is indeed estimated using, as input, the latest neighbours' state received. To this extent, the safety of the control algorithm must be proved against this value. For this kind of delay, our measurements have reported an average value of $70ms$. This ensures that the system behaved in a reliable manner for the entire duration of the experiment. Fig.~\ref{fig:staterttonesingleexperiment} reports the state RTT delay curve in the case of a single experiment. Fig.~\ref{fig:statertt}, on the other hand, aggregates the results of all of the experiments we conducted in the real intersection scenario. In both cases, the outliers can be ignored since, by design, we imposed the rule to discard so-called \emph{late predecessor} packets, i.e., packets with lower sequence numbers arriving at a node that has already successfully received a packet belonging to the same stream and carrying a higher sequence number.}

%Graphs of this delay are reported in Fig.~\ref{fig:appdel1} and Fig.~\ref{fig:appdel2}. The outliers can be ignored as, by design, if with an highest sequence number is received before the predecessor, the late predecessor will be discarded.
\begin{figure}[!t]
    \centering
    \vspace{-10pt}
    \includegraphics[width=0.9\columnwidth]{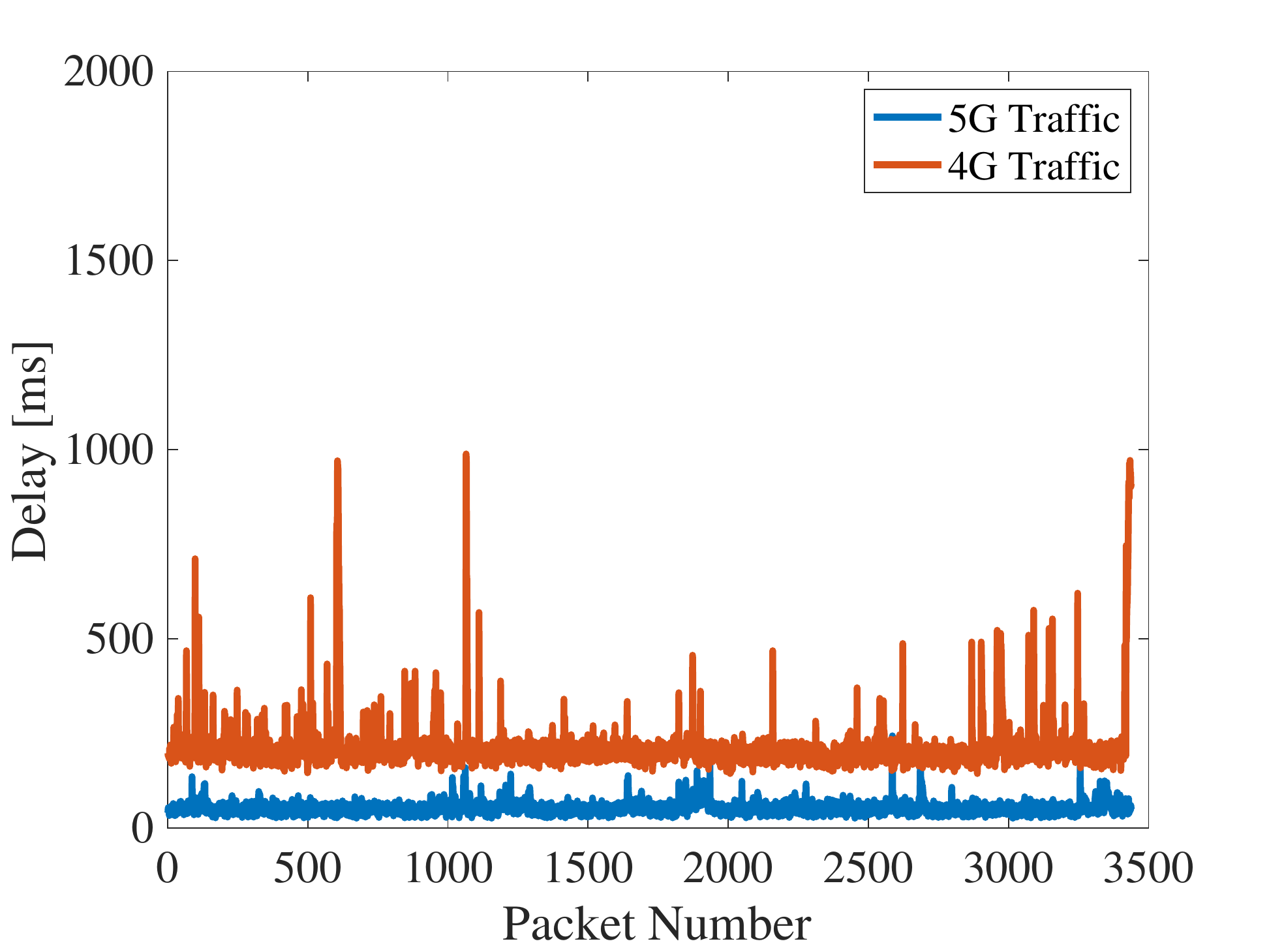}
    \caption{Comparing results with 4G Public cloud: In red it is plotted the trend of delays registered on a public cloud against the timing we collected on our network (blue)}
    \label{fig:appdel1}
        \vspace{-10pt}
\end{figure}

%\subsection{Comparison with an LTE Public network}
\subsection{Further considerations}
As part of our trials, we were also interested in investigating the performance increase deriving from the adoption of a 5G-enabled network using Edge cloud technology. All other things being equal, we ran the same experiments over a publicly reachable LTE infrastructure. Results are reported in Fig.~\ref{fig:appdel1} and clearly show that the measured delay, in case of the public infrastructure, has more than doubled, with an average value that is close to $200ms$. While the observable difference in performance is mostly due to the difference between an application server deployed at the edge and a more distant one, an interesting direction to further explore is to analyze the performance of our framework over a non dedicated network supported by state of the art technology. 

%Eventually, our experiments fully disclose the advantages deriving from the joint utilization of Edge-based Cloud Technologies and LTE in a smart city environment.

For the sake of completeness, we remark that we did not carry out a statistically reliable set of comparative trials. Since we did not get enough LTE data to claim statistically relevant results, the presented graph should be taken with a grain of salt. It nonetheless gives an idea as of the qualitative performance trend in the presence of the two mentioned communication infrastructures.

%The showed comparison is quantitative and we did not get enough LTE data to get statistically relevant results.
%Speaking about delay composition, the final figure is composed of the transport-layer component (associated with the TCP protocol), plus the overhead deriving from the use of the upper-layer websocket protocol, plus the above discussed \emph{State RTT} measured at the level of the application-layer control algorithm.following components + the time of re-transmission (bounded between 0 and 20ms) which is the transmission frequency of the \textit{Traffic Manager}, as it showed in Figure \ref{fig:delcomp}

%\begin{figure}[!ht]
%    \centering
%    \includegraphics[width=.85\columnwidth]{imgs/delays_composition.png}
%    \caption{Delay Composition}
%    \label{fig:delcomp}
%\end{figure}

\subsection{Results from Cooperative Crossing Driving Tests}

%{\bf STEFANIA}

\begin{figure*}[!t]
\centering
\vspace{-10pt}
\subfloat[]{\includegraphics[width=0.32\textwidth]{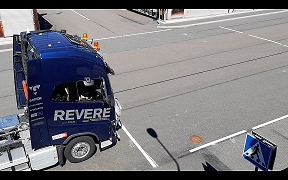}\label{fig:esperimento_1}}
\subfloat[]{ \includegraphics[width=0.32\textwidth]{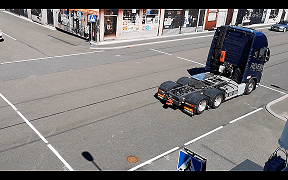}\label{fig:esperimento_2}}
\subfloat[]{\includegraphics[width=0.32\textwidth]{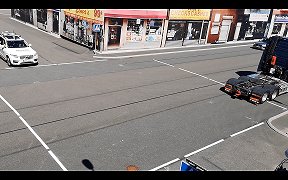}\label{fig:esperimento_3}}\\
\vspace{-10pt}
\subfloat[]{ \includegraphics[width=0.32\textwidth]{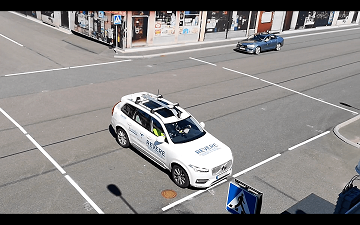}\label{fig:esperimento_4}}
\subfloat[]{\includegraphics[width=0.32\textwidth]{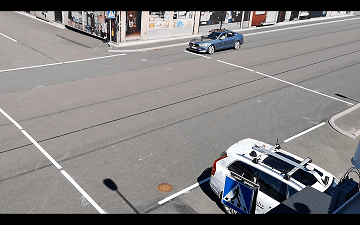}\label{fig:esperimento_5}}
\subfloat[]{ \includegraphics[width=0.32\textwidth]{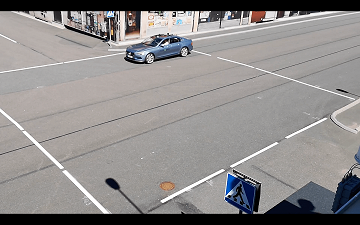}\label{fig:esperimento_6}}\\
\caption{Snapshots of autonomous crossing over 5G}
\vspace{-10pt}
\label{esperimento_incrocio}
\end{figure*}
\begin{figure*}[!t]
\centering
\vspace{-10pt}
\subfloat[]{\includegraphics[width=0.32\textwidth]{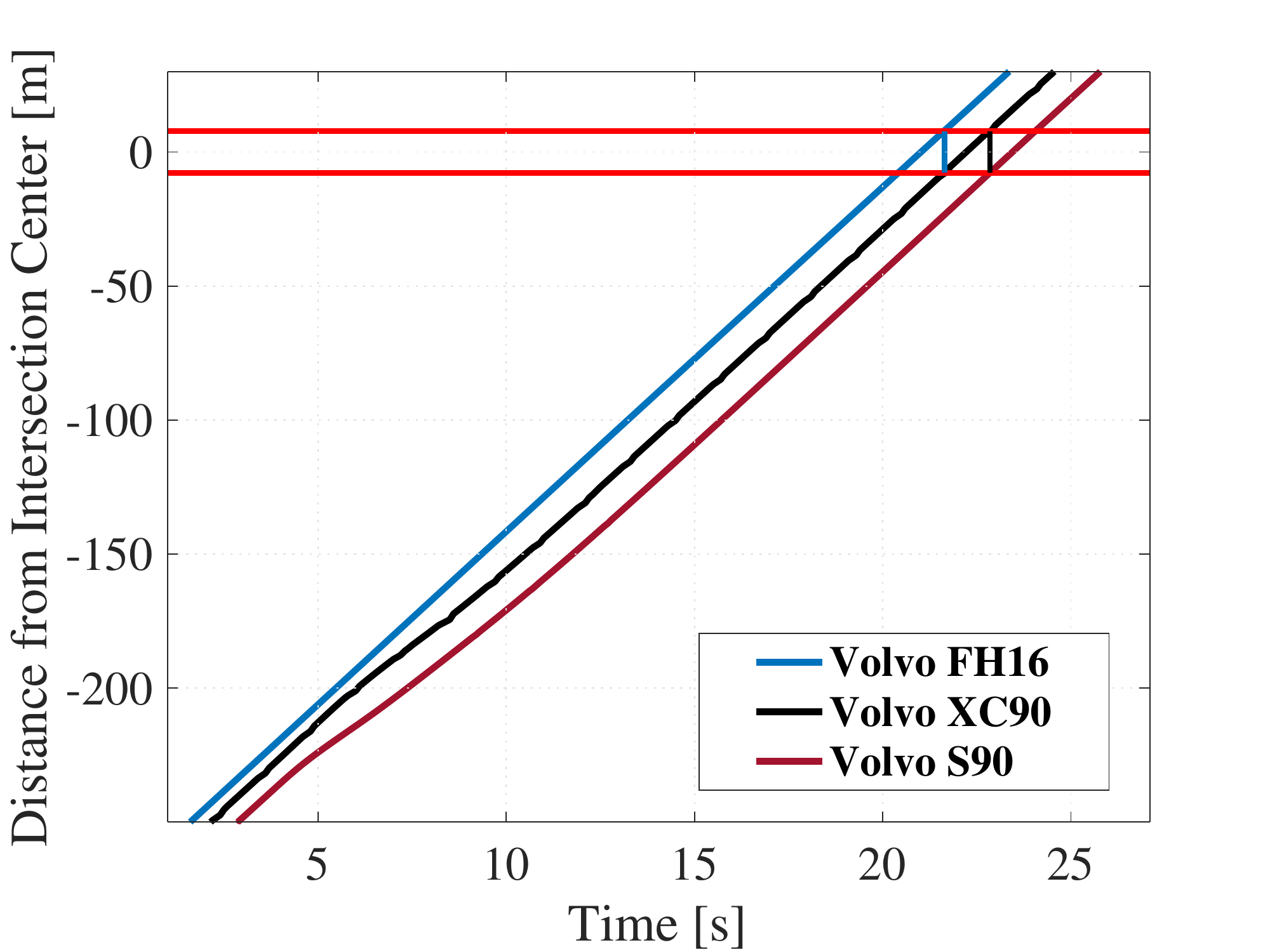}\label{fig:Position}}
\subfloat[]{ \includegraphics[width=0.32\textwidth]{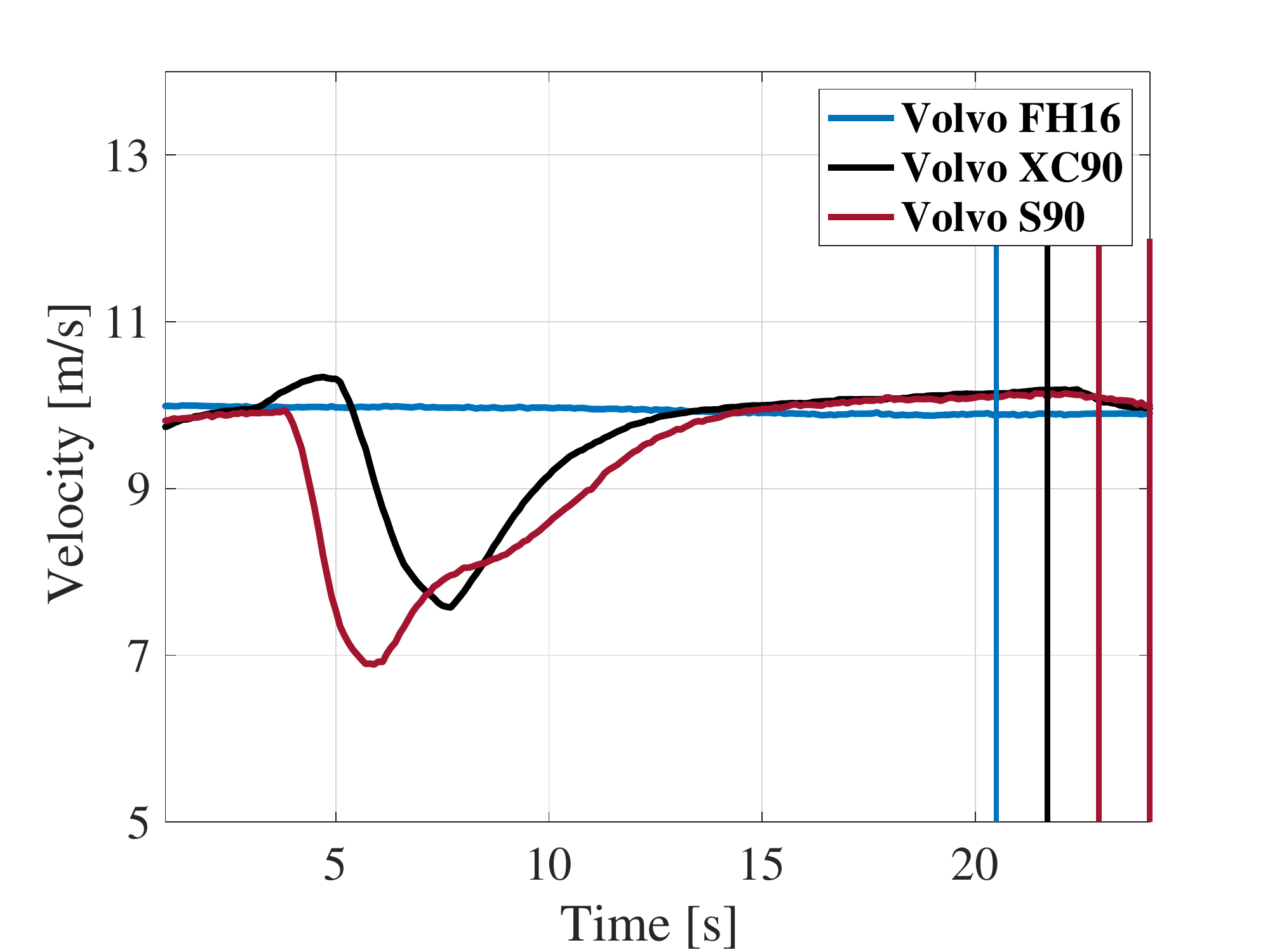}\label{fig:Speed}}
\subfloat[]{\includegraphics[width=0.32\textwidth]{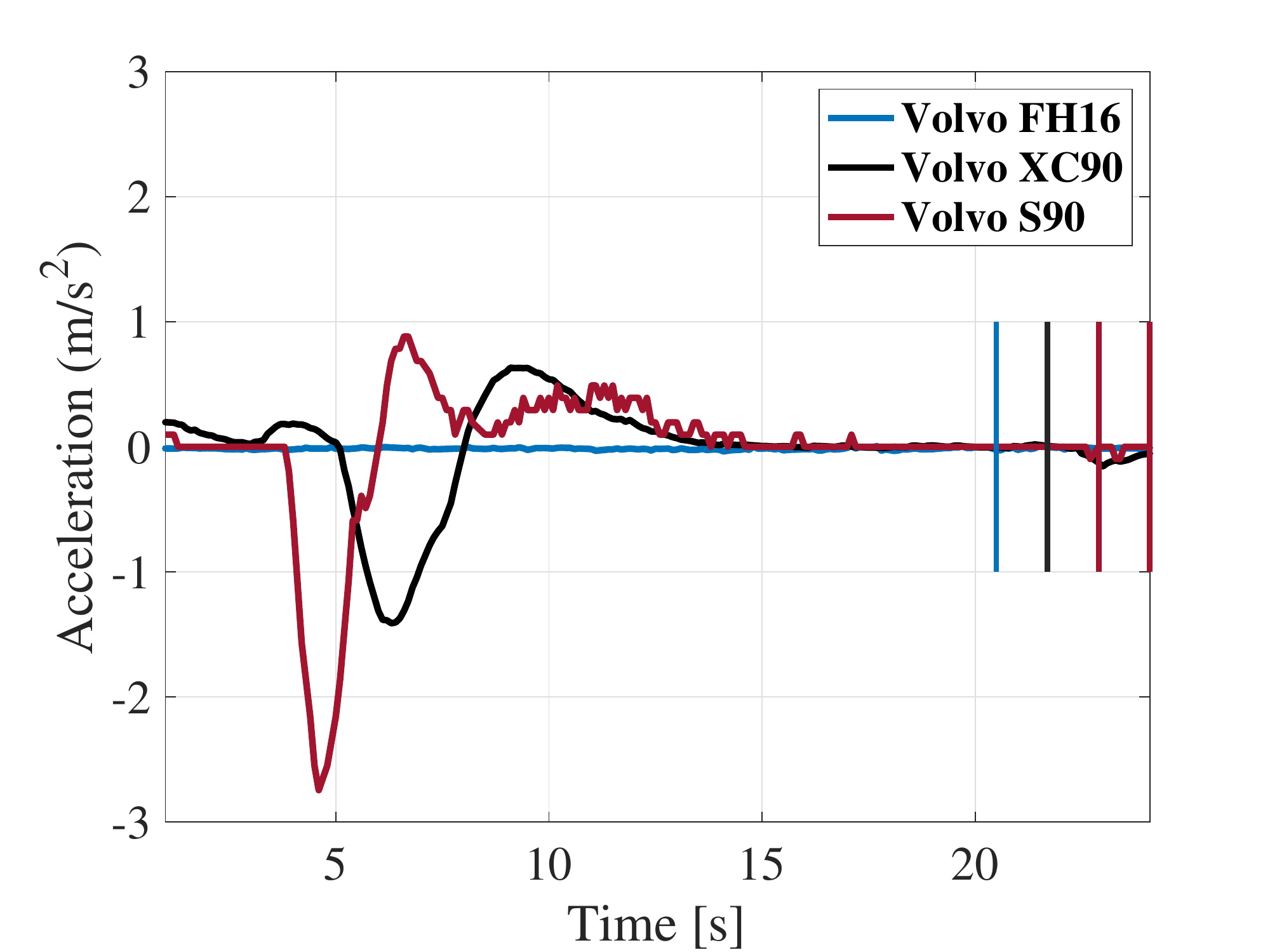}\label{fig:Acceleration}}\\
\vspace{-10pt}
\subfloat[]{ \includegraphics[width=0.32\textwidth]{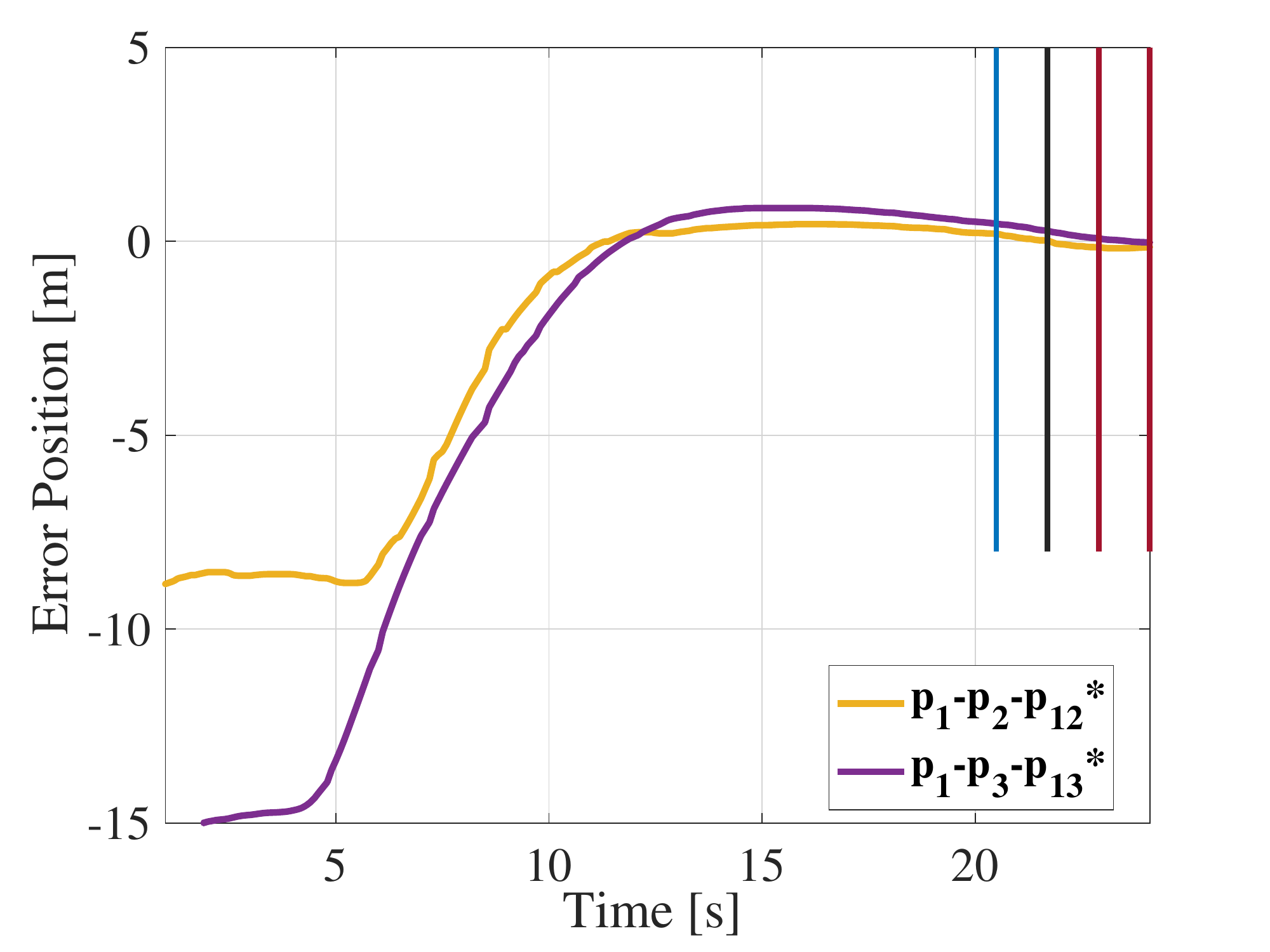}\label{fig:PositionError}}
\subfloat[]{\includegraphics[width=0.32\textwidth]{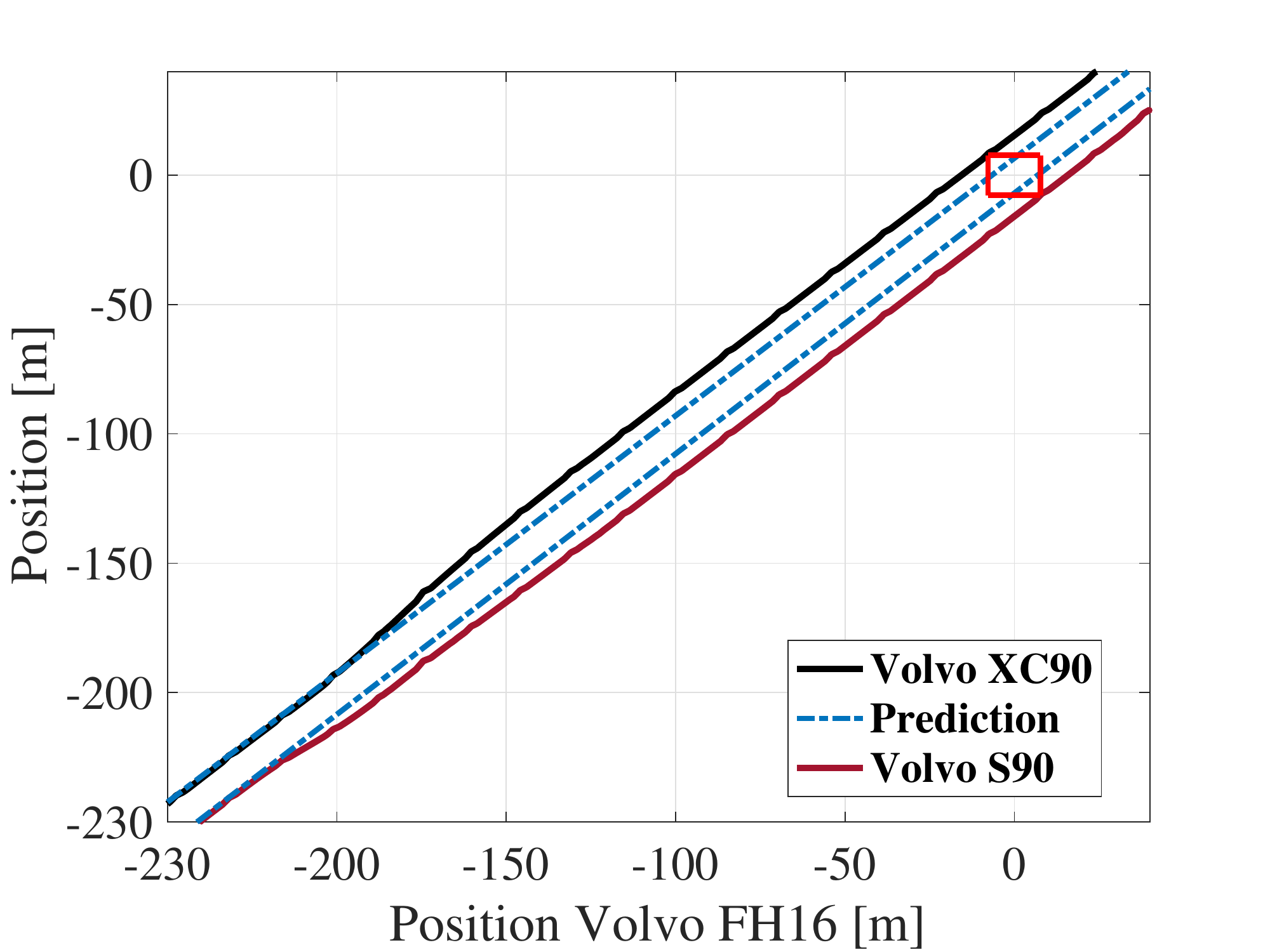}\label{fig:risultatoQuadrato}}
\subfloat[]{ \includegraphics[width=0.32\textwidth]{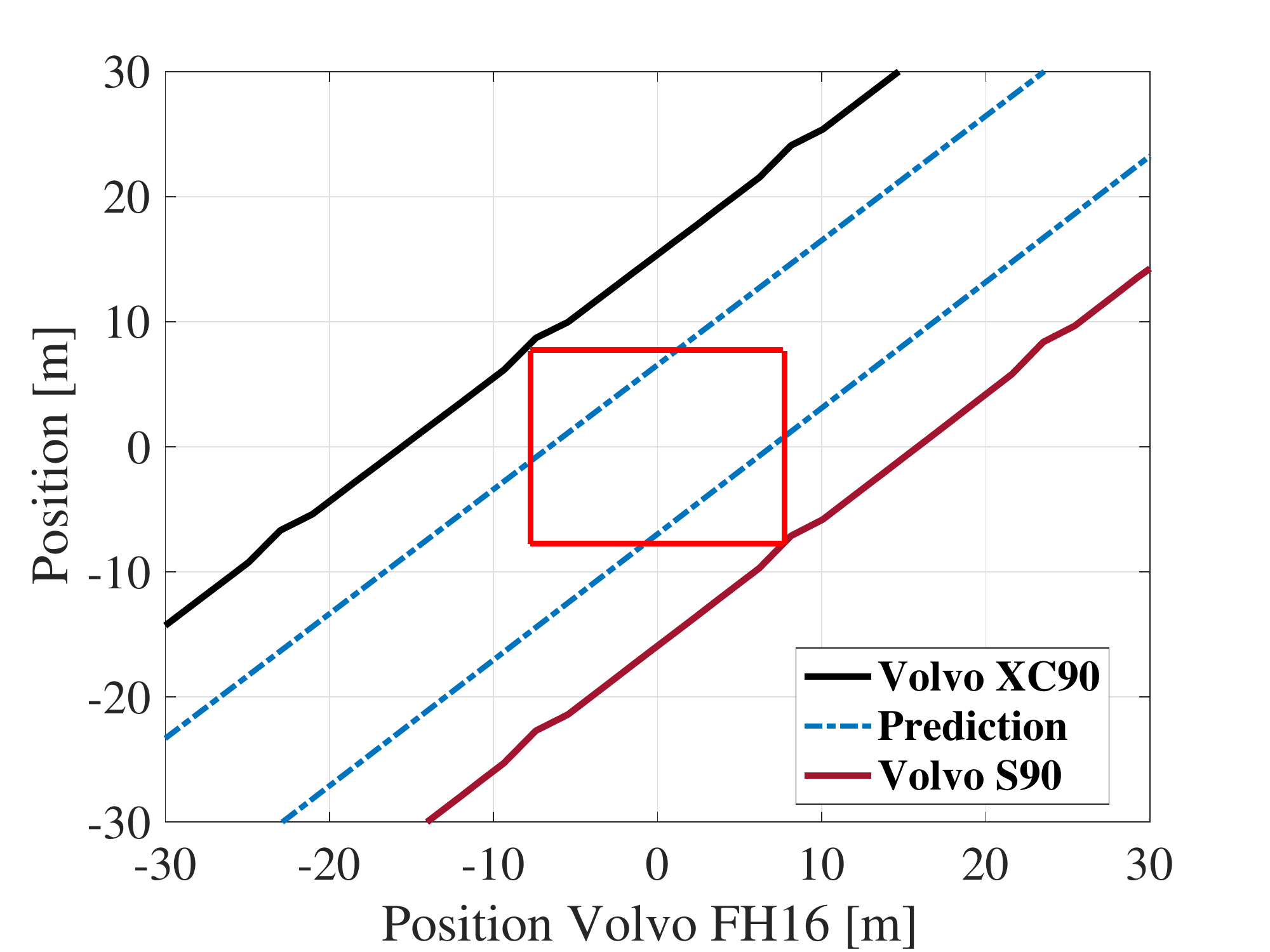}\label{fig:risultatoQuadratoZoom}}\\
\caption{Experimental Results: a) Time history of the position (beginning and end of the CA: solid horizontal lines); b) Time history of the vehicle velocity; c) Time history of the vehicle acceleration; d) Time history of the position errors w.r.t. desired inter-vehicle gaps; e) Position of the $2^{nd}$ vehicle and the $3^{rd}$ vehicle, vs position of the $1^{st}$ vehicle (colliding area: rectangular area; theoretical trajectories without the control correction: dashed line); f) Position of the $2^{nd}$ vehicle and the $3^{rd}$ vehicle, vs position of the $1^{st}$ vehicle: Zoom on the colliding area.
In figures a), b), c) and d) the blue vertical line indicates the time instant at which the 1st vehicle enters the CA; the black vertical line indicates the time instant at which the $1^{st}$ vehicle exits and $2^{nd}$ enters the CA; the bordeaux vertical line indicates the time instant at which the $2^{nd}$ vehicle exits and $3^{rd}$ enters the CA; the second bordeaux vertical line indicates the time instant at which the $3^{rd}$ vehicle exits the CA.}
\vspace{-10pt}
\label{Risultati_esperimenti}
\end{figure*}

Results in \cref{Risultati_esperimenti} show the effectiveness of the finite-time cooperative control protocol in guaranteeing the safe crossing at an intersection, with communication happening over 5G (parameters values characterizing the experimental scenarios are summarized in \Cref{tab:table2}).
Full video of the experiments can be found at \textit{https://www.youtube.com/watch?v=rmjJkIlFMJ4}, while some snapshots are in Fig. \ref{esperimento_incrocio}.
Specifically, the time histories of positions, velocities, and accelerations reported respectively in \cref{fig:Position},  \cref{fig:Speed} and \cref{fig:Acceleration}, confirm that the cooperative control protocol guarantees the exclusive vehicles access into the CA, whose boundaries are indicated with solid red horizontal lines in \cref{fig:Position}.
Indeed, only when the first vehicle has exited the CA, the second vehicle is just ready to enter the intersection (as highlighted by the vertical solid line).
The achievement of the required inter-vehicle spacing correctly matches with the reaching of a common velocity (see \cref{fig:Speed}).
According to the theoretical derivation, the time history of the position error converges to zero with a finite settling time $T$ of about $20 \; [s]$ (see \cref{fig:PositionError}).
Hence, the cooperative algorithm safely converges in a finite time before the first vehicle accesses the CA, hence ensuring collision avoidance at the junction.

To better appreciate the effectiveness of the proposed control approach in avoiding collisions, in \cref{fig:risultatoQuadrato} and \cref{fig:risultatoQuadratoZoom} the position of both the second and the third vehicle, i.e. $p_2(t)$ and $p_3(t)$, are plotted against the position of the first vehicle, i.e. $p_1(t)$.
Here the red square area represents the set of positions for which collision occurs and, as it can be easily observed, both trajectories (solid lines) tangentially touch the critical colliding area, indicated by the red square, so it  follows that, as soon as a vehicle exits the CA, the next one is just ready to enter.
Conversely, the ideal trajectories (dashed-dotted line), i.e. the unsafe trajectories that vehicles would have followed if their initial velocities would have been held without any correction, uncover the occurrence of collisions in the absence of control. 

\begin{table}[!ht]
\footnotesize
  \centering
  \begin{tabular}{lc}
  %  \textbf{Communication system model setting} \\
    \toprule
  \textbf{Parameters}& \textbf{Values}\\
  Positions initial condition $[m]$  & {} \\
   $[p_{1}(0), p_{2}(0), p_{3}(0)]^{\top}$ & $[-220, -235, -250]^{\top}$ \\
Velocities initial condition $[m/s]$ & {} \\
   $[v_{1}(0), v_{2}(0), v_{3}(0)]^{\top}$ & $[10, 9.7, 9.8]^{\top}$ \\
   Control gain $\alpha$ & 0.1   \\
   Vehicle length $L_{i} \; [m]$   & $L_{1}=7.8 \; \; L_{2}=L_{3}= 4.6$ \\
   Headway time $h \; [s]$ & $0.8$ \\
   Distance at standstill $r_{ij} \; [m]$ & $10$\\
    \bottomrule
      {} \\
     \end{tabular}
\caption{EXPERIMENTAL SCENARIO PARAMETERS.}
       \label{tab:table2}
       \vspace{-15pt}
\end{table}

\section{Conclusion}\label{sec:conclusions}
\LMC{In this paper we have successfully modelled, implemented, deployed and tested an autonomous driving system operating in proximity of a street intersection over a pre-5G test network. Experiments have been carried out in urban scenarios where buildings and objects of different shapes and materials act as obstacles against most used sensors, and point to point communication paradigms. The control algorithm has been thought and designed in a distributed way, with the objective of minimizing the impact of transmissions and the computational load on remote devices. Network performance, along with control results, has been analyzed and discussed. }

\section{Acknowledgement}
\LMC{Experiments for this project have been realized in a tight cooperation with Ericsson AB and REVERE Lab (Chalmers). Special thanks to Keerthi Kumar Nagalapur, Henrik Sahlin, Magnus Castell and Anders Aronsson, from the Ericsson team, for the great ideas, the constant help and the frequent discussions. Also, we would like to thank Christian Berger, Arpit Karsolia and Fredrik von Corswant from Revere LAB and Niklas Lundin from AstaZero. Without their contribution and passion, our experiments would have never been possible.}

% if have a single appendix:
%\appendix[Proof of the Zonklar Equations]
% or
%\appendix  % for no appendix heading
% do not use \section anymore after \appendix, only \section*
% is possibly needed

% use appendices with more than one appendix
% then use \section to start each appendix
% you must declare a \section before using any
% \subsection or using \label (\appendices by itself
% starts a section numbered zero.)
%

%\appendices
%\section{Qua potrebbe andarci qualche snapshot di codice particolare}

% use section* for acknowledgment
%\section*{Acknowledgment}
%
%Ringraziamo Ericsson, Chalmers, Arpit e la fredda svezia.

% Can use something like this to put references on a page
% by themselves when using endfloat and the captionsoff option.
\ifCLASSOPTIONcaptionsoff
  \newpage
\fi

% trigger a \newpage just before the given reference
% number - used to balance the columns on the last page
% adjust value as needed - may need to be readjusted if
% the document is modified later
%\IEEEtriggeratref{8}
% The "triggered" command can be changed if desired:
%\IEEEtriggercmd{\enlargethispage{-5in}}

% references section

% can use a bibliography generated by BibTeX as a .bbl file
% BibTeX documentation can be easily obtained at:
% http://mirror.ctan.org/biblio/bibtex/contrib/doc/
% The IEEEtran BibTeX style support page is at:
% http://www.michaelshell.org/tex/ieeetran/bibtex/
%\bibliographystyle{IEEEtran}
% argument is your BibTeX string definitions and bibliography database(s)
%\bibliography{IEEEabrv,../bib/paper}
%
% <OR> manually copy in the resultant .bbl file
% set second argument of \begin to the number of references
% (used to reserve space for the reference number labels box)
% \begin{thebibliography}{1}
%
%\bibitem{IEEEhowto:kopka}
%H.~Kopka and P.~W. Daly, \emph{A Guide to \LaTeX}, 3rd~ed.\hskip 1em plus
%  0.5em minus 0.4em\relax Harlow, England: Addison-Wesley, 1999.
%
%\end{thebibliography}
%\bibliography{bibliography}{}
\bibliographystyle{IEEEtran}
\bibliography{bibliography}

% Generated by IEEEtran.bst, version: 1.14 (2015/08/26)
\begin{thebibliography}{10}
\providecommand{\url}[1]{#1}
\csname url@samestyle\endcsname
\providecommand{\newblock}{\relax}
\providecommand{\bibinfo}[2]{#2}
\providecommand{\BIBentrySTDinterwordspacing}{\spaceskip=0pt\relax}
\providecommand{\BIBentryALTinterwordstretchfactor}{4}
\providecommand{\BIBentryALTinterwordspacing}{\spaceskip=\fontdimen2\font plus
\BIBentryALTinterwordstretchfactor\fontdimen3\font minus
  \fontdimen4\font\relax}
\providecommand{\BIBforeignlanguage}[2]{{%
\expandafter\ifx\csname l@#1\endcsname\relax
\typeout{** WARNING: IEEEtran.bst: No hyphenation pattern has been}%
\typeout{** loaded for the language `#1'. Using the pattern for}%
\typeout{** the default language instead.}%
\else
\language=\csname l@#1\endcsname
\fi
#2}}
\providecommand{\BIBdecl}{\relax}
\BIBdecl

\bibitem{ShadowingEffect}
B.~H. Liu, B.~Otis, S.~Challa, P.~Axon, C.~T. Chou, and S.~Jha, ``On the fading
  and shadowing effects for wireless sensor networks,'' in \emph{2006 IEEE
  International Conference on Mobile Ad Hoc and Sensor Systems}, Oct 2006, pp.
  51--60.

\bibitem{rios2017survey}
J.~Rios-Torres and A.~A. Malikopoulos, ``A survey on the coordination of
  connected and automated vehicles at intersections and merging at highway
  on-ramps,'' \emph{IEEE Transactions on Intelligent Transportation Systems},
  vol.~18, no.~5, pp. 1066--1077, 2017.

\bibitem{zhu2015linear}
F.~Zhu and S.~V. Ukkusuri, ``A linear programming formulation for autonomous
  intersection control within a dynamic traffic assignment and connected
  vehicle environment,'' \emph{Transportation Research Part C: Emerging
  Technologies}, vol.~55, pp. 363--378, 2015.

\bibitem{kamal2015vehicle}
M.~A.~S. Kamal, J.-i. Imura, T.~Hayakawa, A.~Ohata, and K.~Aihara, ``A
  vehicle-intersection coordination scheme for smooth flows of traffic without
  using traffic lights,'' \emph{IEEE Transactions on Intelligent Transportation
  Systems}, vol.~16, no.~3, pp. 1136--1147, 2015.

\bibitem{zheng2017distribute}
Y.~Zheng, S.~E. Li, K.~Li, F.~Borrelli, and J.~K. Hedrick, ``Distributed model
  predictive control for heterogeneous vehicle platoons under unidirectional
  topologies,'' \emph{IEEE Transactions on Control Systems Technology},
  vol.~25, no.~3, pp. 899--910, 2017.

\bibitem{liu2017distributed}
C.~Liu, C.-W. Lin, S.~Shiraishi, and M.~Tomizuka, ``Distributed conflict
  resolution for connected autonomous vehicles,'' \emph{IEEE Transactions on
  Intelligent Vehicles}, 2017.

\bibitem{chalmersexp1}
M.~Zanon, S.~Gros, H.~Wymeersch, and P.~Falcone, ``An asynchronous algorithm
  for optimal vehicle coordination at traffic intersections,'' in
  \emph{Conference: IFAC 2017 World Congress}, 8 2017, [ONLINE AVAILABLE]
  https://www.youtube.com/watch?v=fzkv5beS4uk.

\bibitem{chalmersexp2}
R.~Hult, M.~Zanon, S.~Gros, and P.~Falcone, ``Primal decomposition of the
  optimal coordination of vehicles at traffic intersections,'' in \emph{2016
  IEEE 55th Conference on Decision and Control (CDC)}, 12 2016, [ONLINE
  AVAILABLE] https://www.youtube.com/watch?v=fzkv5beS4uk.

\bibitem{chartrand2006introduction}
G.~Chartrand, \emph{Introduction to graph theory}.\hskip 1em plus 0.5em minus
  0.4em\relax Tata McGraw-Hill Education, 2006.

\bibitem{lu2013finite}
X.~Lu, R.~Lu, S.~Chen, and J.~Lu, ``Finite-time distributed tracking control
  for multi-agent systems with a virtual leader,'' \emph{IEEE Transactions on
  Circuits and Systems I: Regular Papers}, vol.~60, no.~2, pp. 352--362, 2013.

\bibitem{bhat2000finite}
S.~P. Bhat and D.~S. Bernstein, ``Finite-time stability of continuous
  autonomous systems,'' \emph{SIAM Journal on Control and Optimization},
  vol.~38, no.~3, pp. 751--766, 2000.

\bibitem{sun2016finite}
Z.~Sun, S.~Mou, M.~Deghat, and B.~Anderson, ``Finite time distributed
  distance-constrained shape stabilization and flocking control for
  d-dimensional undirected rigid formations,'' \emph{International Journal of
  Robust and Nonlinear Control}, vol.~26, no.~13, pp. 2824--2844, 2016.

\bibitem{coelingh2012all}
E.~Coelingh and S.~Solyom, ``All aboard the robotic road train,'' \emph{Ieee
  Spectrum}, vol.~49, no.~11, 2012.

\bibitem{wu2015distributed}
W.~Wu, J.~Zhang, A.~Luo, and J.~Cao, ``Distributed mutual exclusion algorithms
  for intersection traffic control,'' \emph{IEEE Transactions on Parallel and
  Distributed Systems}, vol.~26, no.~1, pp. 65--74, 2015.

\bibitem{li2006cooperative}
L.~Li and F.-Y. Wang, ``Cooperative driving at blind crossings using
  intervehicle communication,'' \emph{IEEE Transactions on Vehicular
  technology}, vol.~55, no.~6, pp. 1712--1724, 2006.

\bibitem{medina2017cooperative}
A.~I.~M. Medina, N.~van~de Wouw, and H.~Nijmeijer, ``Cooperative intersection
  control based on virtual platooning,'' \emph{IEEE Transactions on Intelligent
  Transportation Systems}, 2017.

\bibitem{slotine1991applied}
J.-J.~E. Slotine, W.~Li \emph{et~al.}, \emph{Applied nonlinear control}.\hskip
  1em plus 0.5em minus 0.4em\relax Prentice hall Englewood Cliffs, NJ, 1991,
  vol. 199, no.~1.

\bibitem{bhat1997finite}
S.~P. Bhat and D.~S. Bernstein, ``Finite-time stability of homogeneous
  systems,'' in \emph{American Control Conference, 1997. Proceedings of the
  1997}, vol.~4.\hskip 1em plus 0.5em minus 0.4em\relax IEEE, 1997, pp.
  2513--2514.

\bibitem{zhao2016distributed}
Y.~Zhao, Z.~Duan, G.~Wen, and G.~Chen, ``Distributed finite-time tracking of
  multiple non-identical second-order nonlinear systems with settling time
  estimation,'' \emph{Automatica}, vol.~64, pp. 86--93, 2016.

\bibitem{wgs84}
NGA, ``Online-available: https://web.archive.org/web/20120402143802/\\
  https://www1.nga.mil/productsservices/geodesyandgeophysics/\\
  worldgeodeticsystem/pages/default.aspx\\ {World Geodetic System website of
  the NGA (archived April 2012)". National Geospatial-Intelligence Agency.
  Archived from the original on April 2, 2012}.''

\bibitem{popovski2017ultra}
P.~Popovski, J.~J. Nielsen, C.~Stefanovic, E.~de~Carvalho, E.~Str{\"o}m, K.~F.
  Trillingsgaard, A.-S. Bana, D.~M. Kim, R.~Kotaba, J.~Park \emph{et~al.},
  ``Ultra-reliable low-latency communication (urllc): principles and building
  blocks,'' \emph{arXiv preprint arXiv:1708.07862}, 2017.

\bibitem{xylomenos1999tcp}
G.~Xylomenos and G.~C. Polyzos, ``Tcp and udp performance over a wireless
  lan,'' in \emph{Proceedings of INFOCOM}.\hskip 1em plus 0.5em minus
  0.4em\relax IEEE, 1999, pp. 439--446.

\bibitem{ws_ietf}
IETF, ``Online-available: https://tools.ietf.org/html/rfc6455rfc6455.''

\bibitem{nodebook}
P.~Marton, \emph{NODE.JS MONITORING, ALERTING \& RELIABILITY 101}.\hskip 1em
  plus 0.5em minus 0.4em\relax RingStack, 2017.

\bibitem{berger2016open}
C.~Berger, ``An open continuous deployment infrastructure for a self-driving
  vehicle ecosystem,'' in \emph{IFIP International Conference on Open Source
  Systems}.\hskip 1em plus 0.5em minus 0.4em\relax Springer, 2016, pp.
  177--183.

\bibitem{berger2017containerized}
C.~Berger, B.~Nguyen, and O.~Benderius, ``Containerized development and
  microservices for self-driving vehicles: Experiences \& best practices,'' in
  \emph{Software Architecture Workshops (ICSAW), 2017 IEEE International
  Conference on}.\hskip 1em plus 0.5em minus 0.4em\relax IEEE, 2017, pp. 7--12.

\bibitem{benderius2018best}
O.~Benderius, C.~Berger, and V.~M. Lundgren, ``The best rated human--machine
  interface design for autonomous vehicles in the 2016 grand cooperative
  driving challenge,'' \emph{IEEE Transactions on Intelligent Transportation
  Systems}, vol.~19, no.~4, pp. 1302--1307, 2018.

\bibitem{abd}
\url{https://www.abdynamics.com/en/products/track-testing/driving-robots/pedal-robots}.

\bibitem{wanninger2004introduction}
L.~Wanninger, ``Introduction to network rtk,'' \emph{IAG Working Group},
  vol.~4, no.~1, pp. 2003--2007, 2004.

\bibitem{merkel2014docker}
D.~Merkel, ``Docker: lightweight linux containers for consistent development
  and deployment,'' \emph{Linux Journal}, vol. 2014, no. 239, p.~2, 2014.

\end{thebibliography}

% biography section
% 
% If you have an EPS/PDF photo (graphicx package needed) extra braces are
% needed around the contents of the optional argument to biography to prevent
% the LaTeX parser from getting confused when it sees the complicated
% \includegraphics command within an optional argument. (You could create
% your own custom macro containing the \includegraphics command to make things
% simpler here.)
%\begin{IEEEbiography}[{\includegraphics[width=1in,height=1.25in,clip,keepaspectratio]{mshell}}]{Michael Shell}
% or if you just want to reserve a space for a photo:

%\begin{IEEEbiography}{Michael Shell}
%Biography text here.
%\end{IEEEbiography}
%
%% if you will not have a photo at all:
%\begin{IEEEbiographynophoto}{John Doe}
%Biography text here.
%\end{IEEEbiographynophoto}
%
%% insert where needed to balance the two columns on the last page with
%% biographies
%%\newpage
%
%\begin{IEEEbiographynophoto}{Jane Doe}
%Biography text here.
%\end{IEEEbiographynophoto}

% You can push biographies down or up by placing
% a \vfill before or after them. The appropriate
% use of \vfill depends on what kind of text is
% on the last page and whether or not the columns
% are being equalized.

%\vfill

% Can be used to pull up biographies so that the bottom of the last one
% is flush with the other column.
%\enlargethispage{-5in}

% that's all folks
\end{document}